\documentclass[preprint,3p]{elsarticle}

\usepackage{stmaryrd}
\usepackage{listings} 
\usepackage{graphicx} 
\usepackage{verbatim} 
\usepackage{subfig}
\usepackage{ifthen}
\usepackage{xspace}
\usepackage{amsmath}
\usepackage{amsthm}
\usepackage{amssymb}
\usepackage{tikz}
\usetikzlibrary{arrows,automata}
\usepackage{url}
\usepackage{gastex}
\usepackage{pstricks,pst-node,pst-tree}
\usepackage{graphics}
\usepackage{mathrsfs} 

\newpsobject{showgrid}{psgrid}{subgriddiv=1,griddots=10,gridlabels=6pt}

\renewcommand{\epsilon}{\varepsilon}
\renewcommand{\star}{*}
\newcommand{\waut}{\ensuremath{\nu}}

\newcommand{\id}[1]{\it #1}

\newcommand{\eps}{\ensuremath{\epsilon}}
\providecommand{\abs}[1]{\lvert#1\rvert}

\newcommand{\prestar}[2]{\mbox{$pre^{*}_{\mathsf{#1}}(#2)$}}

\newcommand{\pcode}[2][\codesize]{
    \fbox{
    \begin{minipage}{0.45\linewidth}
    #1
    \begin{tabbing}
    xx \= xx \= xx \= xx \= xx \= xx \= xx \= xx \= \kill
    #2
    \end{tabbing}
    \end{minipage}
    }
  }

\newcommand{\ppcode}[2][\small]{
\begin{minipage}{0.45\linewidth}
 #1
\begin{tabbing}
xx \= xx \= xx \= xx \= xx \= xx \= xx \= xx \= \kill
#2
\end{tabbing}
\end{minipage}
}

\let\oldmarginpar\marginpar
\renewcommand\marginpar[1]{\-\oldmarginpar[\raggedleft\footnotesize #1]%
    {\raggedright\footnotesize #1}}

\newcommand{\abbrevfullstop}[1]{%
  \ifthenelse{\equal{#1}{.}}{.}{%
    \ifthenelse{\equal{#1}{,} \OR \equal{#1}{;} \OR \equal{#1}{'} %
            \OR \equal{#1}{?} \OR \equal{#1}{!} %
            \OR \equal{#1}{)} \OR \equal{#1}{]} %
            \OR \equal{#1}{~}}{.#1}%
    {.\ #1}}}

\newcommand{\etal}[1]{\emph{et al}\abbrevfullstop{#1}}
\newcommand{\eg}[1]{\emph{e.g}\abbrevfullstop{#1}}
\newcommand{\ie}[1]{\emph{i.e}\abbrevfullstop{#1}}

\newcommand{\tuple}[1]{\langle #1 \rangle}

\newcommand{\bigO}[1]{\mbox{O}(#1)} 
\newcommand{\ra}{\mathtt{a}}
\newcommand{\rb}{\mathtt{b}}
\newcommand{\rc}{\mathtt{c}}

\newcommand{\encode}[1]{\mbox{$\llbracket #1\rrbracket$}}
\newcommand{\nt}[2]{\mbox{$#1_{#2}$}}                      
\newcommand{\readterm}[2]{\mbox{$\mathtt{#1\_at\_#2}$}}   
\newcommand{\writeterm}[2]{\mbox{$\mathtt{set\_#1\_#2}$}} 

\newcommand{\newState}[4]{\node[state,#3](#1)[#4]{#2};}

\newcommand{\newTransition}[4]{\path[->] (#1) edge [#4] node {#3} (#2); }

\newtheorem{example}{Example}
\newtheorem{lemma}{Lemma}
\newtheorem{corollary}{Corollary}
\newtheorem{theorem}{Theorem}
\newtheorem{definition}{Definition}
\newtheorem{proposition}{Proposition}

\newcommand{\sys}[1]{\textsc{#1}}
\newcommand{\oursolver}{\sys{covenant}}
\newcommand{\ignore}[1]{}

\newcounter{proglineno}
\newcommand{\putno}{\refstepcounter{proglineno}\arabic{proglineno}:}

\newcommand{\lang}[1]{\ensuremath{\mathcal{L}(#1)}}

\makeatletter
\def\blfootnote{\xdef\@thefnmark{}\@footnotetext}
\makeatother

\begin{document}
\title{A Complete Refinement Procedure for\\Regular Separability of Context-Free Languages}
\author[unimelb]{Graeme Gange}
\ead{gkgange@unimelb.edu.au}

\author[nasa]{Jorge A. Navas}
\ead{jorge.a.navaslaserna@nasa.gov}

\author[unimelb]{Peter Schachte}
\ead{schachte@unimelb.edu.au}

\author[unimelb]{Harald S{\o}ndergaard}
\ead{harald@unimelb.edu.au}

\author[unimelb]{Peter J. Stuckey}
\ead{pstuckey@unimelb.edu.au}

\address[unimelb]{
  Department of Computing and Information Systems,
  The University of Melbourne,
  Vic. 3010, Australia}
\address[nasa]{
  NASA Ames Research Center,
  Moffett Field,
  CA 94035, USA}

\begin{abstract}
Often, when analyzing the behaviour of systems modelled as 
context-free languages, we wish to know if two languages overlap. 
To this end, we present an effective semi-decision procedure 
for regular separability of context-free languages,
based on counter-example guided abstraction refinement.
We propose two refinement methods, one inexpensive but incomplete,
and the other complete but more expensive.
We provide an experimental evaluation of this procedure, 
and demonstrate its practicality on a range
of verification and language-theoretic instances.

\end{abstract}
\maketitle

\section{Introduction}
\label{sec-intro}

We address the problem of checking whether two given
context-free languages $L_1$ and $L_2$ are disjoint.
This is a fundamental language-theoretical problem.
It is of interest in many practical tasks that call for
some kind of automated reasoning about programs.
This can be
because program behaviour is modelled using context-free languages,
as in software verification approaches that try to capture a program's
control flow as a (pushdown-system) path language.
Or it can be because we wish to reason about string-manipulating
programs, as is the case in software vulnerability detection problems,
where various types of injection attack have to be modelled.

The problem of context-free disjointness is of course undecidable, 
but semi-decision procedures exist for non-disjointness. 
For example, one can systematically generate strings $w$
over the intersection $\Sigma_1 \cap \Sigma_2$, where $\Sigma_1$ is
the alphabet of $L_1$ and $\Sigma_2$ is that of $L_2$.
If some $w$ belongs to both $L_1$ and $L_2$, answer
``yes, the languages overlap.''
It follows that no semi-decision procedure exists for disjointness.
However, semi-decision procedures exist for the
stronger requirement of being separable by a regular language.
For example, one can systematically generate (representations of) 
regular languages over $\Sigma_1 \cup \Sigma_2$, and,
if some such language $R$ is found to satisfy 
$L_1 \subseteq R \land L_2 \subseteq \overline{R}$, answer 
``yes, the languages are disjoint''.

A radically different approach, which we will follow here, 
uses so-called \emph{counter-example guided abstraction refinement
(CEGAR)}~\cite{ClarkeGJLV00} of regular over-approximations.
The scheme is based on repeated approximation refinement, like so:

\begin{enumerate}

\item \emph{Abstraction:} Compute \emph{regular approximations} $R_1$
  and $R_2$ such that $L_1 \subseteq R_1$ and $L_2 \subseteq R_2$.
  (Here $R_1$ and $R_2$ are regular languages, represented using
  regular expressions, say.)

\item \emph{Verification:} Check whether the intersection of $R_1$ and
  $R_2$ is empty using a decision procedure for regular
  expressions. If $R_1 \cap R_2 = \emptyset$ then $L_1 \cap
  L_2 = \emptyset$, so answer ``the languages are disjoint.''
  If $ w \in (R_1 \cap R_2)$, $w \in L_1$, and $w \in L_2$
  then $ L_1 \cap L_2 \neq \emptyset$, 
  so answer ``the languages overlap'' and provide $w$ as a witness.
  Otherwise, go to step 3.

\item \emph{Refinement:} Produce new regular approximations $R_1'$ and
  $R_2'$ such that for each $R_i'$, $i \in \{1,2\}$,  we have $L_i
  \subseteq R_{i}' \subseteq R_i$, and $R'_i \subset R_i$ for some $i$.
  Update the approximations $R_1
  \leftarrow R_1',R_2 \leftarrow R_2'$, and go to step 2.

\end{enumerate}
For the abstraction step, note that regular approximations exist,
trivially.
For the verification step,
we could also take advantage of the 
fact that the class of context-free languages is closed under 
intersection with regular languages; however, this does not eliminate
the need for a refinement procedure.
For the refinement step, note that there is no indication of \emph{how}
the tightening of approximations should be done; indeed
that is the focus of this paper.
The step is clearly well-defined since, if $L \subset R$, there is 
always a regular language $R' \subset R$ such that $L \subseteq R'$.

For a given language $L$ there may well be an infinite chain
$R_1 \supset R_2 \supset \cdots \supset L$ of regular approximations.
This is a source of possible non-termination of the CEGAR scheme.
An interesting question therefore is: 
Are there refinement techniques that can guarantee termination at
least when $L_1$ and $L_2$ are 
\emph{regularly separable} context-free languages, that is, when
there exists a regular language $R$ such that $L_1 \subseteq R$
and $L_2 \subseteq \overline{R}$?%

In this paper we answer this question in the affirmative.
We propose a refinement procedure which can ensure termination of the 
CEGAR-based loop assuming the context-free languages involved are 
regularly separable.
In this sense we provide a refinement procedure which is
\emph{complete} for regularly separable context-free languages.
Of course the question of regular separability of context-free languages 
is itself undecidable~\cite{SzymanskiW73}. 
The method that we propose can also be used on language instances that 
are not regularly separable, and it will often decide such instances 
successfully.
However, in this case, it does not come with a termination guarantee.

\paragraph{Contribution}
The paper rests on regular approximation ideas by 
Nederhof~\cite{Nederhof_chapter1}
and we utilise the efficient $pre^{\star}$ algorithm~\cite{EsparzaRS00}
for intersecting (the language of) a context-free grammar with 
(that of) a finite-state automaton.
We propose a novel refinement procedure for a CEGAR inspired method
to determine whether context-free languages are disjoint, and
we prove the procedure complete for determining regular separability.
In the context of regular approximation, where languages must be 
over-approximated using \emph{regular} languages, separability is 
equivalent to regular separability, so the completeness means
that the refinement procedure is optimal.
On the practical side, the method has important applications in
software verification and security analysis.
We demonstrate its feasibility through an experimental evaluation.

\paragraph{Outline}
Section~\ref{sec-preliminaries} introduces concepts, notation and 
terminology used in the paper.
It also recapitulates relevant results about regular separability and 
language representations.
Section~\ref{sec-refine} proposes a new procedure for regular 
approximation of context-free languages and and shows that the
procedure is complete, in the sense that it proves the
separability for any pair of regularly separable context-free
languages.
Section~\ref{sec-example} provides an example.
In Section~\ref{sec-compare} we place our method in context,
comparing with previously proposed refinement techniques.
In Section~\ref{sec-results} we evaluate the method empirically,
comparing an implementation with the most closely related tool.
Section~\ref{sec-conclusion} concludes.
The appendices contain more peripheral implementation detail and
a description of test cases used for the experimental evaluation.

\section{Preliminaries}
\label{sec-preliminaries}

In this section we recall the notion of regular separability and 
introduce a concept of ``star-contraction'' for regular expressions.

\subsection{Regular and Context-Free Languages}
We first recall some basic formal-language concepts.
These are assumed to be well understood---the only purpose here is to
fix our terminology and notation.
Given an alphabet $\Sigma$, $\Sigma^*$ denotes the set of all finite 
strings of symbols from $\Sigma$.
The string $y$ is a \emph{substring} of string $w$ iff $w = x y z$
for some (possibly empty) strings $x$ and $z$.

The \emph{regular expressions} over an alphabet 
$\Sigma = \{a_1,\ldots,a_n\}$ are $\emptyset$, 
$\epsilon$, $a_1, \ldots a_n$, together with expressions of form
$e_1 | e_2$, $e_1 \cdot e_2$, and $e^*$, where $e$, $e_1$ and $e_2$
are regular expressions.
Here $|$ denotes union, $\cdot$ denotes language 
concatenation, and ${}^*$ is Kleene star.
As is common, we will often omit $\cdot$, so that juxtaposition of
$e_1$ and $e_2$ denotes concatenation of the corresponding languages.
Given a finite set $E = \{e_1,\ldots,e_k\}$ of regular expressions, 
we let $\bigparallel E$ stand for the regular expression 
$e_1 | \cdots | e_k$ 
(in particular, $\bigparallel \emptyset = \emptyset$).
We let $Reg_{\Sigma}$ denote the set of regular expressions over
alphabet $\Sigma$.

A (non-deterministic) \emph{finite-state automaton} is a quintuple
$\tuple{Q,\Sigma,\delta,q_0,F}$
where $Q$ is the set of states, $\Sigma$ is the alphabet, $\delta$ is
the transition relation, $q_0$ is the start state, and $F$ is the set
of accept states.
The presence of $(q,x,q')$ in $\delta \subseteq Q \times \Sigma \times Q$
indicates that, on reading symbol $x$ while in state $q$,
the automaton may proceed to state $q'$.
If $\delta$ is a total function, that is, 
if for all $q \in Q, x \in \Sigma,
|\{q' \mid (q,x,q') \in \delta\}|=1$,
then the automaton is \emph{deterministic}.

A language which can be expressed as a regular expression
(or equivalently, has a finite-state automaton that recognises it)
is \emph{regular}.
The language recognised by automaton $A$ is written as $\lang{A}$.
Similarly, $\lang{e}$ is the language denoted by regular expression $e$.

A \emph{context-free grammar}, or CFG, is a quadruple
$G = \tuple{V,\Sigma,P,S}$, where $V$ is the set of 
variables (non-terminals), $S$ is the start symbol, and $P$ is the
set of productions (or rules).
Each production is of form $X \rightarrow \alpha$ with 
$X \in V$ and $\alpha \in (V \cup \Sigma)^*$.
If $X \rightarrow \alpha$ is a production in $P$ then, for all
$\beta, \gamma \in (V \cup \Sigma)^*$, we say that
$\beta X \gamma$ \emph{yields} $\beta \alpha \gamma$, written
$\beta X \gamma \Rightarrow \beta \alpha \gamma$.
The language \emph{generated} by $G$ is 
$\lang{G} = \{w \in \Sigma^* \mid S \Rightarrow^* w\}$, 
where $\Rightarrow^*$ is the reflexive transitive
closure of $\Rightarrow$.
A set of strings is a context-free language (CFL) iff it is generated by 
some CFG.

In algorithms we represent regular languages using finite-state automata,
and CFLs using CFGs.
When there is little risk of confusion, we ignore the distinction between
a language and its representation.
Hence we may, for example, apply set operations to 
(the transition relation of) a finite-state automaton.

\subsection{Regular Separability}
As our approach uses regular approximations to context-free languages,
we cannot hope to prove separation for arbitrary disjoint pairs of
context-free languages. 
Instead we focus on pairs of of \emph{regularly separable} languages.

\begin{definition}[Regularly separable] \rm
Two context-free languages $L_1$ and $L_2$ are \emph{regularly separable}
iff there exists a regular language $R$ such that $L_1 \subseteq R$
and $L_2 \subseteq \overline{R}$ where $\overline{R}$ is the
complement of $R$.
\end{definition}
\noindent
It will be useful to have a slightly different view of separability:
\begin{definition}[Separating pair] \rm
Given a pair $(L_1,L_2)$ of context-free languages, 
a pair $(R_1, R_2)$ of regular languages form
a \emph{separating pair} for $(L_1,L_2)$ iff 
$L_1 \subseteq R_1$, $L_2 \subseteq R_2$,
and $R_1 \cap R_2 = \emptyset$.
\end{definition}
\begin{lemma}
Context-free languages $L_1$ and $L_2$ are regularly separable
iff there exists some separating pair $(R_1, R_2)$ for $(L_1, L_2)$.
\end{lemma}
\begin{proof}\rm
If $(R_1, R_2)$ is a separating pair then $L_1 \subseteq R_1$,
and $L_2 \subseteq R_2 \subseteq \overline{R_1}$, 
so $L_1$ and $L_2$ are regularly separable.
Conversely, if the regular language $R$ separates $L_1$ and $L_2$
then we have $L_1 \subseteq R$, $L_2 \subseteq \overline{R}$,
and $R \cap \overline{R} = \emptyset$. 
So $(R, \overline{R})$ is a separating pair.
\end{proof}
To see that there are disjoint context-free languages that are not 
regularly separable, 
consider $L = \{\mathsf{a}^n \mathsf{b}^n \mid n \geq 0\}$.
Both $L$ and $\overline{L}$ are non-regular context-free languages,
and therefore not regularly separable.
The problem of checking whether a pair of context-free languages is
regularly separable is undecidable~\cite{hunt82-grammar}.

\subsection{Star-Contraction}\label{sec-contract}

\begin{definition}[Union-free regular language] \rm
  A regular expression is \emph{union-free} iff it does not
  use the union operation.
  A regular language is \emph{union-free regular}
  if it can be written as a union-free regular expression.
  We use $\id{Reg}'_\Sigma$ to denote the set of union-free
  regular expressions.
\end{definition}
\noindent
Union-free regular languages are also known as
\emph{star-dot regular} languages~\cite{brzozowski69-stardot}.

\begin{definition}[Union-free decomposition] \rm
  A union-free decomposition~\cite{nagy04-regular} of
  a regular language $R$ is a finite
  set of union-free regular languages $R_1, \ldots, R_n$
  such that $R = R_1 \cup \ldots \cup R_n$.
\end{definition}

\begin{theorem}[Nagy~\cite{nagy04-regular}]\label{thm-nagy}
  Every regular language $R$ admits some finite union-free
  decomposition.
\end{theorem}

\noindent
Theorem~\ref{thm-nagy} is not surprising; 
it utilises the well-known equivalence $(r_1 | r_2)^* = (r_1^* r_2^*)^*$.

The following concept is central to this paper's ideas.
For a given union-free language, it is
convenient to consider particular sets of sub-languages:

\begin{definition}[Star-contraction]\label{def-contract} \rm
  The star-contraction $\kappa(P)$ of a union-free
  regular expression is the set of languages obtained by replacing 
  some subset of $\star$-enclosed subterms
 in $P$ with $\epsilon$.
  The naive construction is confounded by the presence of nested
  $\star$ operators, as distinct portions of the subterms may
  occur in each outer repetition. 
  the star-contraction is defined as:
  \[
    \begin{array}{rcll}
      \kappa(\ra) &=& \{\ra\} & \mbox{for $\ra \in \Sigma \cup \{\eps\}$} \\
      \kappa(e^*) &=& \{ (\bigparallel E)^\star \mid E \subseteq \kappa(e) \} \\
      \kappa(e_1 \cdot e_2) &=& \{r_1 \cdot r_2 \mid r_1 \in \kappa(e_1) \land
			r_2 \in \kappa(e_2)\} 
    \end{array}
  \]
\end{definition}
\begin{example}
  Consider the union-free regular expression 
  $e = (\ra \rb^\star \rc^\star)^\star$.
  The star-contraction of the parenthesised term is
  $\kappa(\ra \rb^\star \rc^\star) = \{ \ra, \ra \rb^\star, \ra \rc^\star, \ra \rb^\star \rc^\star \}$.

  The star-contraction of $e$ (after elimination of equivalent languages) is then:
    \[\kappa(e) = \{ \eps, \ra^\star, (\ra \rb^\star)^\star, (\ra \rc^\star)^\star,
    (\ra \rb^\star | \ra \rc^\star)^\star, (\ra \rb^\star \rc^\star)^\star \}.\]
  Note how the elements $r \in \kappa(e)$ with $r \not= e$ make particular
  subsets of ${\cal L}(e)$ explicit.
  For example, $\ra^\star$ is the subset that makes no use of $\rb$ 
  or $\rc$,
  whereas $(\ra \rb^\star|\ra \rc^\star)^\star$ is the set of words 
  in which $\rb$ and $\rc$ are not adjacent.
\end{example}

\noindent
We later use $\kappa(R)$, that is, $\kappa$ applied to a regular language,
to denote $\kappa(e_1) \cup \ldots \cup \kappa(e_m)$, where $\{ e_1, \ldots, e_m \}$ is
some (arbitrary) union-free decomposition of $R$.
Note that the star-decomposition of a regular language is not unique; different union-free decompositions
give rise to different star-contractions. We assume, for a regular language $R$,  $\kappa(R)$ deterministically returns
\emph{some} valid star-contraction of $R$.

$\kappa(R)$ has several properties which will be useful in the following:
\begin{proposition}
  For any regular language $R$:
  \begin{enumerate}
    \item $\kappa(R)$ is finite.
    \item For every regular expression 
	$e \in \kappa(R)$, $\lang{e} \subseteq R$.
    \item $\bigcup \{ \lang{e} \mid e \in \kappa(R) \} = R$
  \end{enumerate}
\end{proposition}

\section{Refining Regular Abstractions}
\label{sec-refine}
We now describe the main idea behind the refinement phase.  
We are interested in the intersection of a finite set of languages,
but without loss of generality, we consider the intersection of 
just two context-free languages $L_1$ and $L_2$ 
(provided as context-free grammars).

We assume a decision procedure that
returns ``no'' if $\lang{A_1} \cap \lang{A_2} = \emptyset$ or returns a
witness $w$ if $ w \in \lang{A_1} \cap \lang{A_2} \neq \emptyset$, 
where $A_1$ and $A_2$ are finite-state automata recognising regular 
languages $R_1$ and $R_2$, respectively 
(that is, $\lang{A_1} = R_1$ and $\lang{A_2} = R_2$).
Moreover, our refinement procedure will require the solving of constraints
of the form $ A = A_1 \setminus A_2$ where $A$, $A_1$ and $A_2$ are
finite-state automata, that is, $A$ recognises 
$\lang{A_1} \cap \overline{\lang{A_2}}$.

Assume that at some point we have regular approximations 
$R_1 \supseteq L_1$ and $R_2 \supseteq L_2$, 
and we have found some witness $w$ such that
$w \in R_1 \cap R_2$, but $w \notin L_1 \cap L_2$. 
There are three cases to consider:
\[
\setlength{\arraycolsep}{2ex}
\begin{array}{ccc}
  (1) ~~ w \notin L_1 \wedge w \in L_2
& (2) ~~ w \in L_1 \wedge w \notin L_2
& (3) ~~ w \notin L_1 \wedge w \notin L_2
\end{array}
\]

\noindent
For cases (1) and (2) we should refine $R_1$ and $R_2$, respectively.
For case (3) we could choose to refine either $R_1$ or $R_2$, or both. 
In our implementation, we always refine all the regular approximations.

If $w \notin L_i$ then a straightforward refinement is to
produce a new abstraction $R_i \setminus \{ w \}$ in place of $R_i$. 
However, this refinement process will
rarely converge, as we can exclude only finitely many strings in
finite time.
We must instead formulate a refinement procedure which generalizes a counterexample
to an infinite set of words.

\subsection{Star-generalizations}\label{sec-stargen}
The refinement procedure in this section operates by taking the regular expression recognizing
the single counterexample $w$, and progressively augmenting it with $\star$ operators while
ensuring the counterexample and query remain disjoint.

\begin{definition}[Star-generalization]\label{def-Xi} \rm
  The star-generalizations of a word $w$ is the
  set $\Xi$ of regular expressions given by
    \[
      \begin{array}{ccl}
        \Xi(\eps) &=& \{ \eps \} \\
        \Xi(x) &=& \{ x, x^\star \} ~\mbox{for $x \in \Sigma$} \\
        \Xi(\alpha_1 \ldots \alpha_n) &=& \{\alpha_1 \ldots \alpha_n, (\alpha_1 \ldots \alpha_n)^\star\} \cup
          \left\{ e_1 e_2, (e_1 e_2)^\star ~\middle|~
          \begin{array}{c}
            e_1 \in \Xi(\alpha_1 \ldots \alpha_i), \\
            e_2 \in \Xi(\alpha_{i+1} \ldots \alpha_n), \\
            i \in [1, n-1]
          \end{array}
          \right
        \}
      \end{array}
    \]
\end{definition}
A star-generalization of $w$ is a language which can be constructed by adding (nested)
unbounded repetition of intervals in $w$.

We shall represent a star-generalization as a pair $\langle w, S \rangle$
of a word $w = w_1 \ldots w_n$ and the set $S$ of ranges
 covered by $\star$-operators. We denote such a range by the pair $(i, j)$ of
 the first and last index covered by the operator.
$S$ must satisfy the constraints
\begin{equation}\label{eq:nest}
  \begin{array}{c}
    \forall (i, j) \in S~.~i, j \in [0, n], i < j \\
    \forall (i, j), (i', j') \in S~.~ j \le i' \vee j' \le j
  \end{array}
\end{equation}
This ensures the set of $\star$-enclosed ranges are well-formed.
We shall use $\lang{\langle w, S\rangle}$ to denote the language 
that results from such a generalization.

\begin{definition}[Star-generalization with respect to a language]
  The star-generalizations of $w$ with respect to some language $L$
  (denoted $\Xi_{L}(w)$)
  is the set of star-generalizations of $w$ which are contained in $L$.
  Formally, this can be expressed as
  \[ \Xi_{L}(w) = \{ r ~|~ r \in \Xi(w),~\lang{r} \subseteq L \}. \]
\end{definition}
\begin{definition}[Maximal star-generalization] \rm
  A maximal star-generalization of $w$ with respect to some language $L$
  is a star-generalization $r$ of $w$ such that
  there is no other star-generalization $r'$ with $\lang{r} \subset \lang{r'} \subseteq L$.
\end{definition}
A maximal star-generalization is not necessarily unique.
Consider generalizing $ab$ with respect to $(\ra^\star{}\rb|\ra\rb^\star)$ --
both $\ra^\star{}\rb$ and $\ra\rb^\star$ are incomparable maximal generalizations.

A greedy procedure for constructing a star-generalization
is given in Figure~\ref{pcode:refine}.
The algorithm takes as input a witness $w \in R_1 \cap R_2$, 
context-free language $L$ such that $w \notin L$, and $L$'s regular 
approximation $A$ (either $R_1$ or $R_2$).
The procedure begins with a trivial star-generalization recognizing
only $w$. $P$ stores the set of $(i, j)$ pairs where a $\star$-operation
may be introduced without causing the generalization to be malformed.
At each step, we add one of the candidate operations to the generalization,
then remove any pairs from $P$ which are no longer feasible (because they
violate the nestedness requirement).
Following (\ref{eq:nest}), this is the set of pairs $(i', j')$ such
that $i < i' < j < j' \vee i' < i < j' < j$.

\setcounter{proglineno}{0}
\begin{figure}[t]
  \centerline{
    \pcode{
      \> \textsf{refine}($w, L, A$) \\
       \putno \> \> let $w$ be $x_1 \cdot x_2 \cdots x_n$ \\
       \putno \> \> $S$ := $\emptyset$ \\
       \putno \> \> $P$ := $\{ (i,j) ~|~ i,j \in [0,n], i < j \}$ \\
       \putno \> \> \label{intp-begin}\textbf{while} $P \not= \emptyset$ \\
       \putno \> \> \> \label{gen}\textrm{choose} $(i,j) \in P$ \\
       \putno \> \> \> $S'$ := $S \cup \{(i, j)\}$ \\
       \putno \> \> \> $P$ := $P \setminus \{(i, j)\}$ \\
       \putno\label{test} \> \> \> \textbf{if} $(\lang{\langle w, S' \rangle } \cap L = \emptyset)$ \\
       \> \> \> \> $S$ := $S'$ \\
       \> \> \> \> $P$ := $\left\{ (i', j') ~\middle|~
          \begin{array}{ll}
                  & (i', j') \in P \\
           \wedge & (j \le i' \vee j' \le j) \\
           \wedge & (j' \le i \vee j \le j')
        \end{array} \right\}$ \\
      \putno \> \> \textbf{return} $A \setminus \lang{\langle w, S \rangle}$
    }
  }
  \caption{\label{pcode:refine}
    Refining a regular approximation greedily by removing a maximal
    star-generalization of some counterexample $w$.
  }
\end{figure}

\begin{figure}
  \centerline{
    \begin{tabular}{ccccc}
    \includegraphics[width=0.25\linewidth]{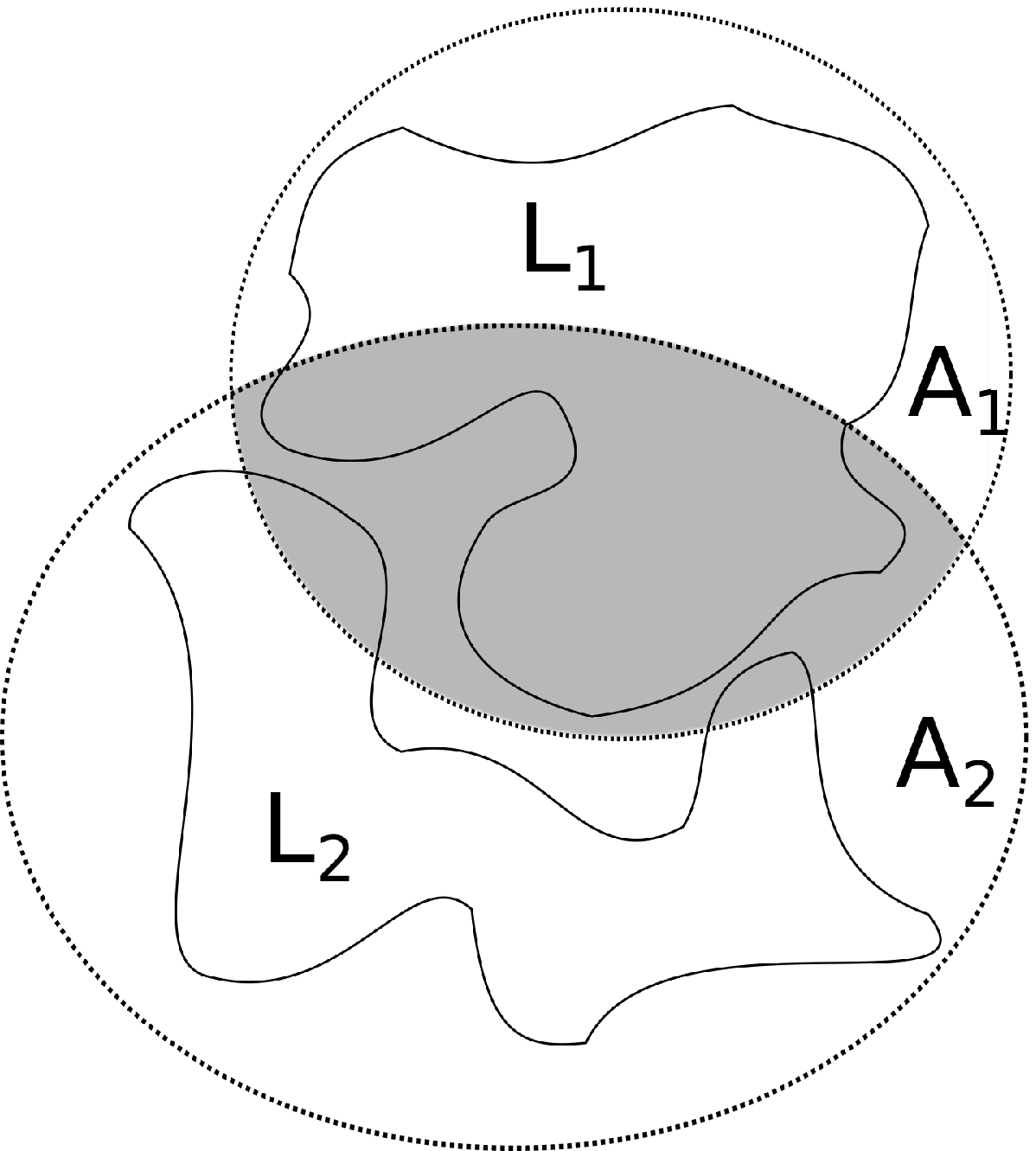}
    & \qquad &
    \includegraphics[width=0.25\linewidth]{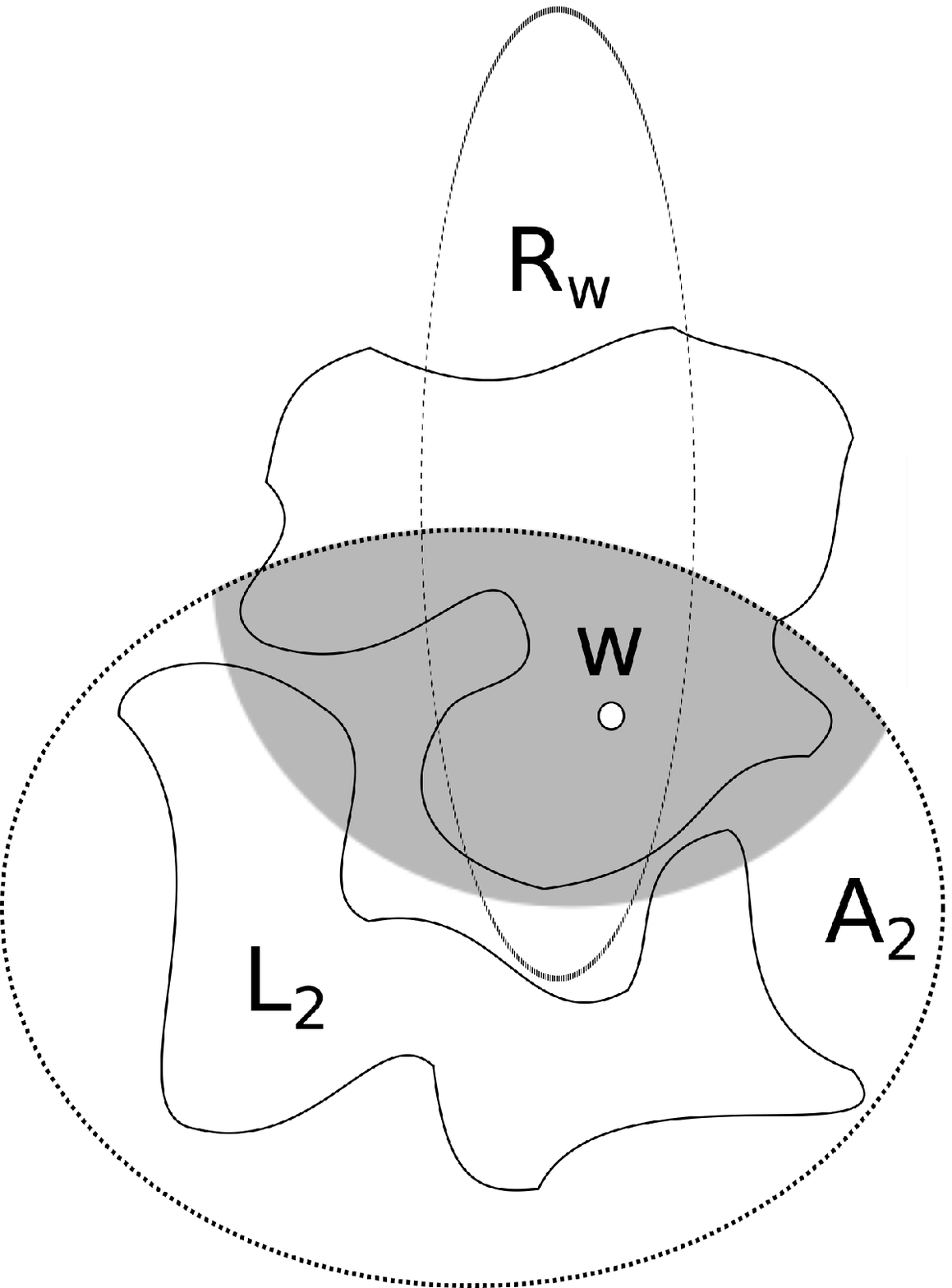}
    & \qquad &
    \includegraphics[width=0.25\linewidth]{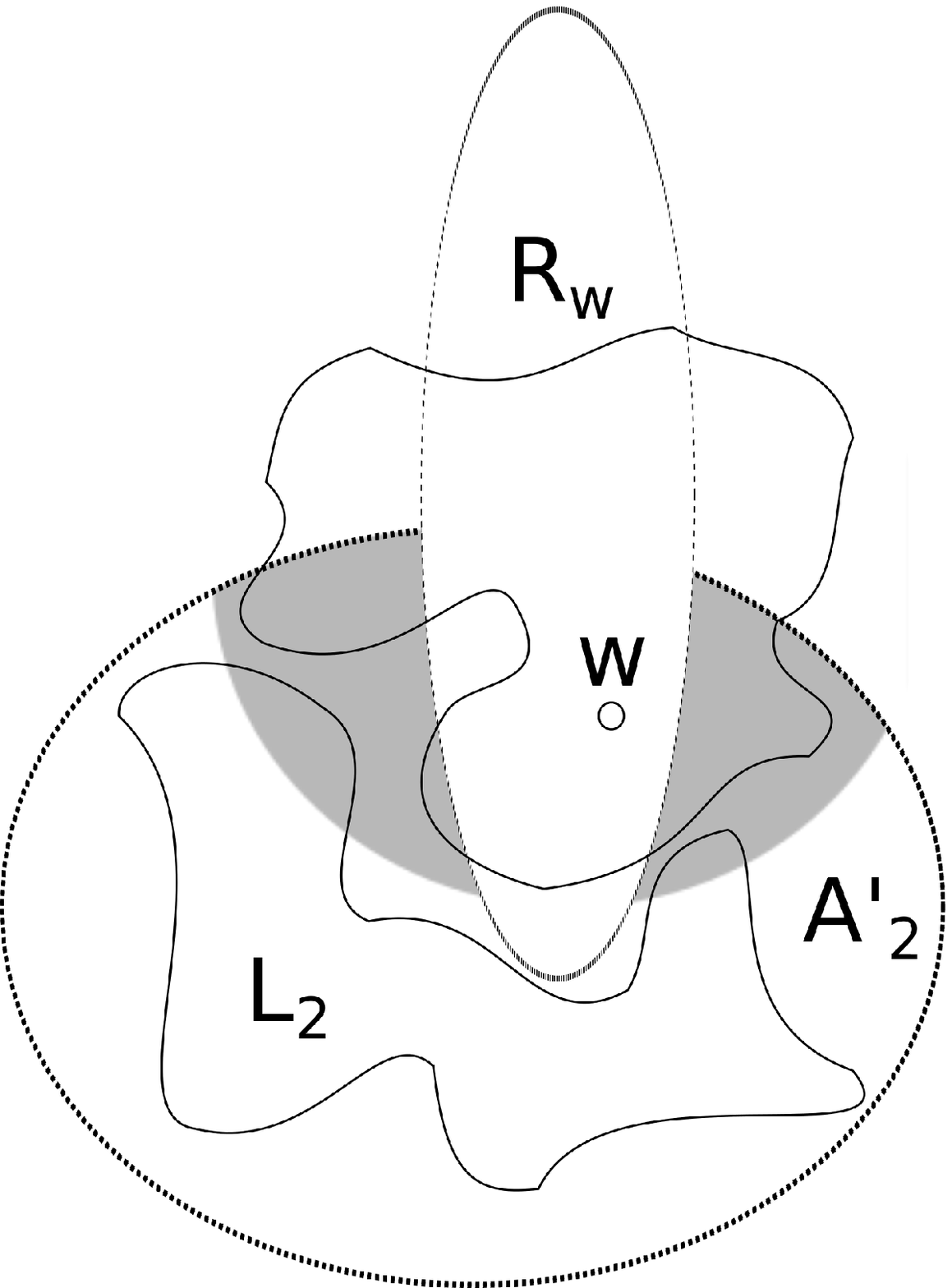} \\
    (a) & & (b) &  & (c)
    \end{tabular}
  }
  \caption{
    (a) A pair of context free languages $L_1, L_2$ and their initial approximations $A_1, A_2$.
        The intersection of the approximations is shaded.
    (b) A counterexample $w \in A_1 \cap A_2$ has been identified. 
        As $w \notin L_2$, we build a generalization
        $R_w$ such that $\{w\} \subseteq R_w \subseteq \overline{L_2}$. 
    (c) $L_2$'s new approximation is $A'_2 = A_2 \setminus R_w$.
    \label{fig:overview}
  }
\end{figure}

It is worth pointing out that this refinement procedure is \emph{anytime}:
If for some reason it would seem necessary or advantageous,
one can, without compromising correctness,
interrupt the while loop having considered only a subset of the possible
$\star$-augmentations, thereby settling for a smaller generalization.

\begin{example}\label{ex-stargen} \rm
Let 
$L = \{\mathtt{a}^i \mathtt{b}^{i+1} \mid i \geq  0\}$ be
(currently) approximated by $\ra^*\rb^*$,
and let the witness $w \not\in L$ be $\ra \ra \rb$.
To refine the regular approximation, 
$\mathsf{refine}(w,L,\ra^*\rb^*)$
begins with the trivial star-generalization $\tuple{w, \emptyset}$.
We greedily augment the counterexample with $\star$-operations,
in this case following lexicographic order,
with the accumulated language here described with a regular expression:

\[
\begin{array}{|cl|}
   \hline
   \star\mbox{-augmentation} & \multicolumn{1}{c|}{\mbox{$\lang{\langle w, S \rangle}$}}
\\ \hline
         & \ra \ra \rb
\\ (0,1) & \mbox{include, obtaining $\ra^\star \ra \rb$}
\\ (1,2) & \mbox{exclude, as $\rb \in L$}
\\ (2,3) & \mbox{exclude, as $\ra \rb \rb \in L$}
\\ (0, 2) & \mbox{exclude, as $\rb \in L$}
\\ (1, 3) & \mbox{include, obtaining $\ra^\star (\ra \rb)^\star$}
\\ (0, 3) & \mbox{include, obtaining $(\ra^\star (\ra \rb)^\star)^\star$}
\\ \hline
\end{array}
\]

\vspace{1ex}
\noindent
The complement of the resulting language is $(\ra^\star\ra\rb)^\star\rb(\ra|\rb)^\star$.
The language returned by $\mathsf{refine}(w,L,\mathtt{a}^*\mathtt{b}^*)$
is (represented by) the intersection automaton of this and the 
initial approximation
$\mathtt{a}^*\mathtt{b}^*$, yielding $\rb \rb^\star | \ra \ra^\star \rb \rb \rb^\star$.
This is the new, improved, regular approximation of $L$.
By construction, it does not contain $\mathtt{aab}$, 
but more importantly, along with $\mathtt{ab}$,
infinitely many other strings have been discarded from
the previous approximation $\mathtt{a}^*\mathtt{b}^*$,
namely all those strings containing at least one $\ra$ and
exactly one $\rb$.
\end{example}

The overall flow of the refinement step is illustrated in Figure~\ref{fig:overview}.
Figure~\ref{pcode:refine}'s refinement procedure has these properties:

\begin{enumerate}
  \item it is \emph{sound}: $L_i \subseteq (R_i \setminus \lang{R})
  \subseteq R_i$.

\item it \emph{terminates}: the while loop will be executed
  at most $\frac{n(n+1)}{2}$ times.

\item it is \emph{progressive}: the same witness $w$ cannot be
  produced again upon successive calls to \textsf{refine}.
\end{enumerate}

The refinement algorithm involves a check (line~\ref{test}) to see if
a regular and a context-free language overlap.  This problem is
decidable in polynomial time\footnote{The algorithm \emph{pre$^{*}$}
  described in~\cite{EsparzaR97,EsparzaRS00} has a time complexity of
  $O(|R|\times|Q|^{3})$ and space complexity of $O(|R|\times|Q|^{2})$,
  where $|R|$ is the number of productions in the context-free grammar
  and $|Q|$ is the number of states in the automaton. We use this
  algorithm in our implementation.}.  

\newcommand{\recon}{\ensuremath{\mathsf{r}}}
The following theorem is critical to the completeness result,
showing the relationship between $\kappa$ and $\Xi$.
Note that here we are intersecting sets of languages,
rather than the languages themselves.
\begin{lemma}\label{lemma-reconstruct_0}
  Let $e$ be a union-free regular expression, and $\kappa$ and $\Xi$ be the star-contraction and star-generalizations
  given in Definitions~\ref{def-contract} and \ref{def-Xi}.
  Then for all $w \in \lang{e}$, $\kappa(e) \cap \Xi(w) \neq \emptyset$.
\end{lemma}
\begin{proof}
  Assume $e = \{\ra\}, \ra \in \Sigma \cup \{\eps\}$. Then $w = \ra$.
  From the definitions, we have $\{\ra\} \in \kappa(e)$, and $\{\ra\} \in \Xi(w)$.

  Assume $e = e_1 \cdot e_2$, such that the induction hypothesis holds on $e_1$ and $e_2$.
  Consider some word $w \in \lang{e}$. We can partition $w$ into $w_1 \cdot w_2$, such that
  $w_1 \in \lang{e_1}$, $w_2 \in \lang{e_2}$.
  By the induction hypothesis, there is some $\recon_1$, $\recon_2$ such that
  $\recon_1 \in \kappa(e_1) \cap \Xi(w_1)$, $\recon_2 \in \kappa(e_2) \cap \Xi(w_2)$.
  As $\recon_1 \in \kappa(e_1)$ and $\recon_2 \in \kappa(e_2)$,
  from Definition~\ref{def-contract} we have
  $\recon_1 \cdot \recon_2 \in \kappa(e)$. By Definition~\ref{def-Xi}, we also have
  $\recon_1 \cdot \recon_2 \in \Xi(w)$. Therefore $\recon_1 \cdot \recon_2 \in \kappa(e) \cap \Xi(w)$.

  Assume $e = e'^\star$, for some $e'$ satisfying the induction hypothesis.
  Consider some word $w \in e$. We can partition $w$ into $w_1 \ldots w_k$, such
  that each $w_i \in e'$.
  By the induction hypothesis, each $w_i$ admits some star-generalization
  $\recon_i \in \kappa(e') \cap \Xi(w_i)$.
  Consider the generalization $\recon$ given by
    \[ \recon = (\recon(w_1)^\star \ldots \recon(w_k)^\star)^\star .\]
  By Definition~\ref{def-Xi}, $\recon \in \Xi(w)$.
  $\recon$ is equivalent to $(\recon(w_1) \cup \ldots \cup \recon(w_k))^\star$, which
  is in $\kappa(e)$. Therefore, $\recon \in \Xi(w) \cap \kappa(e)$.
\end{proof}

\begin{theorem}\label{lemma-reconstruct}
  Let $R$ be a regular language, and $\kappa$ and $\Xi$ be the star-contraction and star-generalizations
  given in Definitions~\ref{def-contract} and \ref{def-Xi}.
  Then for all $w \in R$, $\kappa(R) \cap \Xi(w) \neq \emptyset$.
\end{theorem}
\begin{proof}
  As noted in Section~\ref{sec-contract}, the star-contraction of a regular language $R$
  is computed based on some union-free decomposition $E = \{e_1, \ldots, e_n\}$ of $R$.
  Consider $w \in R$. Therefore, there is some $e \in E$ such that $w \in \lang{e}$.
  By Lemma~\ref{lemma-reconstruct_0}, $\kappa(e) \cap \Xi(w) \neq \emptyset$.
  As $\kappa(e) \subseteq \kappa(R)$, we have $\kappa(R) \cap \Xi(w) \neq \emptyset$.
\end{proof}

\subsection{Epsilon-generalization}\label{sec-epsgen}
The notion of star-generalization, while useful for
reasoning, does not integrate well into existing automaton
algorithms. In this section, we introduce a slightly different
form of generalization.

Let $\waut(w)$ denote the automaton recognizing the single word 
$w \ x_1 \cdots x_n$. 
We have $\waut(x_1 \ldots x_n) = \tuple{Q, \Sigma, \delta, q_0, \{q_n\}}$, with
$Q = \{q_0, \ldots, q_n\}$ and $\delta = \{ (q_{i-1}, x_i, q_i) ~|~ i \in [1, n] \}$.
\begin{definition}[Epsilon-generalization]\label{def-epsgen}
  An epsilon-generalization of $w$ is any language obtained by
  augmenting the transition function of $\waut(w)$ with
  additional edges $E \cup P$, given by:
  \[ E \subseteq \{ (q_i, \eps, q_j) ~|~ i < j \} \]
  \[ P \subseteq \{ (q_{j-1}, w_j, q_{i}) ~|~ i < j \} \]
\end{definition}
That is, we may only introduce epsilon-transitions forwards; backwards transitions
always consume the same input character as the original outgoing transition from
the source state.
\newcommand{\Gen}{\ensuremath{\mathsf{Gen}}}
\newcommand{\gen}{\ensuremath{\mathsf{g}}}
We shall use $\Gen(w)$ to denote the set of epsilon-generalizations of $w$.
Note that where $\Xi(w)$ is a set of languages, $\Gen(w)$ is a set of automata.
Similarly, $\Gen_{L}(w)$ denotes the set of epsilon-generalizations with respect
to some language $L$.

\begin{lemma}\label{lemma-concat}
  Let $A_1 = \tuple{Q_1, \Sigma, \delta_1, q_{10}, \{q_{1n}\}}$
  and $A_2 = \tuple{Q_2, \Sigma, \delta_2, q_{20}, \{q_{2n}\}}$
  be (nondeterministic) finite-state automata such that $q_{1n}$
  has no outgoing edges.

  Then the automaton
  $A_1 \circ A_2 = \tuple{(Q_1 \cup Q_2) \setminus \{q_{1n}\},
      \Sigma, \delta_1 \cup \delta_2 [q_{1n} \mapsto q_{20}], \{q_{2n}\}}$
  recognizes the language $\lang{A_1} \cdot \lang{A_2}$.
\end{lemma}
\begin{proof}
  Let $A = A_1 \circ A_2$.
  Assume $w \in \lang{A_1} \cdot \lang{A_2}$. 
  Then $w = w_1 \cdot w_2$, for some
  $w_1 \in \lang{A_1}, w_2 \in \lang{A_2}$.
  So there is some path from $q_{10}$ to $q_{20}$ matching $w_1$ in
  $A_1 \circ A_2$, and a path from $q_{20}$ to $q_{2n}$ matching $w_2$.
  Therefore $w$ is recognized by $A$.

  Assume $w$ is recognized by $A$. There are no transitions
  from states in $Q_1$ to states in $Q_2$ except $q_{20}$; therefore,
  any path from $q_{10}$ to $q_{2n}$ in $A$ must pass through $q_{20}$.
  There are no transitions from states in $Q_2$ to states
  in $Q_1$ (as $q_{1n}$ had no outgoing edges). Therefore, once
  reaching a state in $Q_2$, a path through $A$ must remain in $Q_2$.

  Hence we can divide the path through $A$ into a prefix,
  following transitions exclusively in $A_1$ and reaching $q_{20}$,
  and a suffix from $q_{20}$ to $q_{2n}$ following transitions
  exclusively in $A_2$. Therefore, $w \in \lang{A_1} \cdot \lang{A_2}$.
\end{proof}

\begin{theorem}\label{thm-stareps}
  For any word $w$, $e \in \Xi(w)$, there is some $A \in \Gen(w)$ such that
  $\lang{e} = \lang{A}$.
  That is, star-generalization of $w$ may be expressed by an
  equivalent epsilon-generalization.
\end{theorem}
\begin{proof}
  Consider $e \in \Xi(w)$.
  
  Assume $e = \ra$, $\ra \in \Sigma \cup \{\eps\}$.
  $\waut(\ra) \in \Gen(w)$, and $\lang{\waut(\ra)} = \lang{e}$.
  Also, the final state of $\waut(\ra)$ has no outgoing transitions.

  Assume that for all expressions
  $e'$ of up to depth $k$, if $e' \in \Xi(w)$ there is some
  $A \in \Gen(w)$ such that $\lang{A} = \lang{e'}$, and the accept
  state of $A$ has no outgoing transitions.
  Consider some star-generalization $e \in \Xi(w)$ of depth $k + 1$.

  Assume $e = e'^\star$ for some generalization $e' \in \Xi(w)$.
  As $e$ has depth $k+1$, $e'$ is of depth $k$.
  Then there is some automaton
  $A = \tuple{Q, \Sigma, \delta, q_0, \{q_n\}} \in \Gen(w)$ such that $\lang{A} = \lang{e'}$.
  We construct a new automaton $A'$ with transition relation $\delta'$
  \[ \delta' = \delta \cup \{q_0 \xrightarrow{\eps} q_n\}
    \cup \{ (q_{j-1} \xrightarrow{w_j} q_{0}) ~|~ (q_j \xrightarrow{\eps^\star} q_n) \in \delta \} \]
  The added transitions are of the form permitted by Definition~\ref{def-epsgen}.
  As $A$ is an epsilon-generalization of $w$, $A'$ is also a valid epsilon-generalization.
  As the accept state of $A$ had no outgoing transitions, and we have not added any transitions
  beginning at $q_n$, the accept state of $A'$ also has no outgoing transitions.
  We now consider the language recognized by $A'$.
  Assume some word $w$ is in $\lang{e'^\star}$. Then either $w = \eps$, or $w = w_1 \ldots w_m$ such that
  $w_i \in \lang{e'} \setminus \{\eps\}$. If $w = \eps$, then $w$ is recognized by $A'$ (by $q_0 \xrightarrow{\eps} q_n$).
  Otherwise, each $w_i$ is recognized by some path $p_i$ from $q_j$ to $q_k$ in $A$, such that $q_n$ is reachable
  from $q_k$ by $\eps$-transitions. Let $q_{k'}$ be the second-last state in $p_i$. By construction, there must
  be some alternate transition from $q_{k'}$ to $q_{0}$ in $A'$;
  then there must be some path following $w_i$ from $q_0$ to $q_0$ in $A'$. Therefore, $w$ is recognized by
  $A'$.
  Now assume there is some word $w \neq \eps$ recognized by $A'$. We can partition $w$ into sub-words $w_1 \ldots w_m$
  such that the path of each $w_i$ with $i < m$ starts at $q_0$, makes its final transition via an introduced
  edge, and uses no other introduced edges. As the path corresponding to $w_i$ finishes with an introduced edge,
  there must be some corresponding path from $q_0$ to $q_n$ in $A$. So $w_i \in \lang{e'}$ for $i < m$. And as the
  path corresponding to $w_m$ starts at $q_0$ and does not use any introduced edges, $w_m \in \lang{e'}$.
  Therefore, $w \in \lang{e'^\star} = \lang{e}$.
  Therefore $\lang{A'} = \lang{e'^\star} = \lang{e}$.

  Assume $e = e_1 \cdot e_2$ for $e_1 \in \Xi(w_1)$, $e_2 \in \Xi(w_2)$ and $w = w_1 w_2$.
  As $e_1$ and $e_2$ have depth at most $k$, there exists $A_1 \in \Gen(w_1)$ and
  $A_2 \in \Gen(w_2)$ satisfying the induction hypothesis.
  The automaton $A' = A_1 \circ A_2$ (with $\circ$ as defined in Lemma~\ref{lemma-concat})
  is a valid epsilon-generalization. So $A \in \Gen(w)$.
  By Lemma~\ref{lemma-concat}, $A' = A_1 \circ A_2$ recognizes $\lang{A_1} \cdot \lang{A_2} = \lang{e_1} \cdot \lang{e_2}$.
  Therefore, $\lang{A} = \lang{e}$.
  As the final state of $A_2$ had no outgoing transitions, and we have not added any outgoing transitions
  from $q_n$, the final state of $A'$ has no outgoing transitions.

  As we can construct epsilon-generalizations for trivial star generalizations (depth $k=1$), epsilon-generalizations
  of size $k$ can be constructed from star generalizations of size $k-1$, any element of $\Xi(w)$
  must correspond to an equivalent element of $\Gen(w)$.
\end{proof}

We can incrementally construct an epsilon-generalization with respect to
some co-context-free language $\overline{L}$ by adapting algorithms for
the intersection of context-free languages and finite automata, such as
the $\mathit{pre}^{*}$ algorithm described in~\cite{EsparzaR97}.

Essentially, we maintain a table $\tau \subseteq Q \times \Gamma \times Q$
such that $(q_i, P, q_j) \in \tau$ iff the context-free production $P$
can be generated by some sub-word recognized by a path from $q_i$ to $q_j$.
Whenever a new transition is added to the generalization, we update $\tau$
with any newly feasible productions; if the start production $S$ is ever
generated on a path from $q_0$ to $q_n$, the generalization has ceased
to be valid -- in which case we revert the table and discard the most
recent augmentation.
The $\mathit{pre}^{*}$ algorithm, as described in~\cite{EsparzaRS00},
exhibits $\bigO{\abs{\Gamma} \abs{Q}^3}$ worst-case time complexity.
In the worst case, where every generalization step fails after the maximum
number of steps, this gives the generalization procedure a worst-case
complexity of $\bigO{\abs{\Gamma} \abs{Q}^5}$.

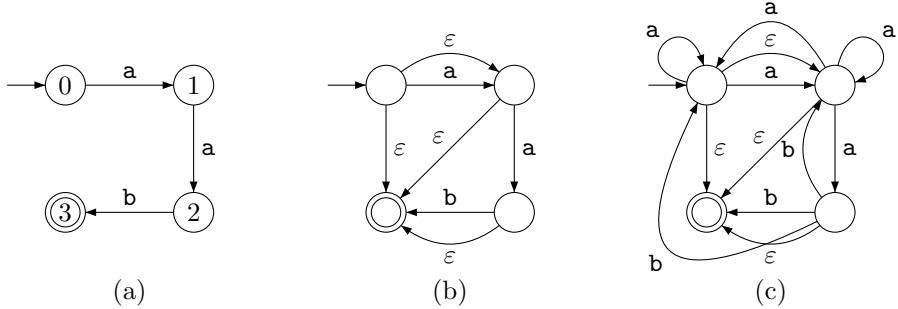
\begin{figure}[t]
\begin{center}
  \unitlength=3pt
  \begin{picture}(120,40)(0,-10)
  \gasset{Nw=5,Nh=5,curvedepth=0}
  \node[Nmarks=i,iangle=180](L0)(0,16){0}
  \node(L1)(16,16){1}
  \node(L2)(16,0){2}
  \node[Nmarks=r](L3)(0,0){3}
  \drawedge[ELside=l](L0,L1){\tt a}
  \drawedge[ELside=l](L1,L2){\tt a}
  \drawedge[ELside=r](L2,L3){\tt b}

  \gasset{Nw=5,Nh=5,curvedepth=0}
  \node[Nmarks=i,iangle=180](C0)(40,16){~}
  \node(C1)(56,16){~}
  \node(C2)(56,0){~}
  \node[Nmarks=r](C3)(40,0){~}
  \drawedge[ELside=l](C0,C1){\tt a}
  \drawedge[ELside=l](C1,C2){\tt a}
  \drawedge[ELside=r](C2,C3){\tt b}
  \drawedge[ELside=l](C0,C3){\eps}
  \drawedge[ELside=r](C1,C3){\eps}
  \gasset{curvedepth=4}
  \drawedge[ELside=l](C0,C1){\eps}
  \drawedge[ELside=l](C2,C3){\eps}

  \gasset{Nw=5,Nh=5,curvedepth=0}
  \node[Nmarks=i,iangle=180](C0)(80,16){~}
  \node(C1)(96,16){~}
  \node(C2)(96,0){~}
  \node[Nmarks=r](C3)(80,0){~}
  \drawedge[ELside=l](C0,C1){\tt a}
  \drawedge[ELside=l](C1,C2){\tt a}
  \drawedge[ELside=r](C2,C3){\tt b}
  \drawedge[ELside=l](C0,C3){\eps}
  \drawedge[ELside=r](C1,C3){\eps}
  \drawloop[loopangle=135,loopdiam=5](C0){\tt a}
  \drawloop[loopangle=45,loopdiam=5](C1){\tt a}
  \gasset{curvedepth=4}
  \drawedge[ELside=l](C0,C1){\eps}
  \drawedge[ELside=l](C2,C3){\eps}
  \drawedge[ELside=l](C2,C1){\tt b}
  \gasset{curvedepth=-8}
  \drawedge[ELside=r](C1,C0){\tt a}
  \gasset{curvedepth=18}
  \drawedge[ELside=l](C2,C0){\tt b}

  \node[Nframe=n](labelA)(8,-10){(a)}
  \node[Nframe=n](labelB)(48,-10){(b)}
  \node[Nframe=n](labelC)(88,-10){(c)}
  \end{picture}
\end{center}
\caption{Construction of $\mathrm{refine}_{\eps}(w,L)$, 
where $w = \ra\ra\rb$ and $L = \{\mathtt{a}^i \mathtt{b}^{i+1} \mid i \geq 0\}$:
(a) Initial automaton for $w$;
(b) the automaton after adding forward $\epsilon$-transitions;
(c) the final epsilon-generalization.\label{fig:example-ab}}
\end{figure}

\begin{example}\label{ex-epsgen} \rm
  Consider again the star-generalization of $\ra\ra\rb$ described in Example~\ref{ex-stargen}.
  If we were instead constructing an epsilon-generalization, we start with the automaton
  $A$ recognizing $w$, shown in Figure~\ref{fig:example-ab}(a). 

  We first try greedily adding forwards epsilon transitions. This yields the
  automaton shown in Figure~\ref{fig:example-ab}(b), corresponding to the language
  $(\ra? (\ra \rb?)?)?$.

  The first backward transition we attempt is $(q_{1-1}, \ra, q_0)$, followed by
  $(q_{2-1}, \ra, q_{1})$. The next transition, $(q_{3-1}, \rb, q_{2})$, cannot
  be added as it would accept $\ra \rb \rb$, which is in $L$.
  This process continues, resulting in the final automaton shown in Figure~\ref{fig:example-ab}(c).
  The language recognized by this automaton is $(\ra^\star \ra \rb)^\star \ra^\star)$, which is
  equivalent to the language obtained by star-generalization in Example~\ref{ex-stargen}.
\end{example}

\subsection{Maximum generalization}\label{sec-maxgen}
The procedure described in the previous section constructs \emph{some}
maximal element of $\Xi_{L}(w)$ (or $\Gen_{L}(w)$).
It is, however, undirected; the generalization is chosen blindly from the set of
possible maximal generalizations.

Even if the query languages are regularly separable, it is possible that
the refinement step may choose an infinite sequence of generalizations
which, though maximal, cannot separate the queries.

We can instead construct a generalization $\mathsf{gen}(L, w)$ which
computes the \emph{union} of all maximal star-generalizations of
$w$ with respect to $\overline{L}$. That is, it computes
$\xi_{\overline{L}}(w) = \bigcup \Xi_{\overline{L}}(w)$ directly.
We shall refer to this as the \emph{maximum} star-generalization.
A possible (though inefficient) method for computing this is given in
Figure~\ref{pcode:maxgen}.

\setcounter{proglineno}{0}
\begin{figure}[t]
  \centerline{
    \pcode{
      \> \textsf{maxgen}($L, w$) \\
       \putno \> \> let $w$ be $x_1 \cdot x_2 \cdots x_n$ \\
       \putno \> \> $S$ := $\emptyset$ \\
       \putno \> \> $P$ := $\{ (i, j) ~|~ i, j \in [0, n], i \le j \}$ \\
       \putno \> \> \textbf{return} $\mathsf{gen}(L, \langle w, S \rangle, P)$ \\
       \\
     \> \textsf{maxgen}($L, \langle w, S \rangle, \emptyset$) \\
     \putno \> \> \textbf{return} $\lang{\langle w, S \rangle}$ \\
     \\
     \> \textsf{maxgen}($L, \langle w, S \rangle, \{(i, j)\} \cup P$) \\
       \putno \> \> $R_f$ := \textsf{gen}($L, \langle w, S \rangle, P$) \\
       \putno \> \> $S'$ := $S \cup \{(i, j)\}$ \\
       \putno\label{gen-safe-check} \> \> \textbf{if} ($L \cap \lang{\langle w, S'\rangle} = \emptyset$) \\
       \putno\> \> \> $P'$ := $\left\{ (i', j') ~\middle|~
          \begin{array}{ll}
                  & (i', j') \in P \\
           \wedge & (j \le i' \vee j' \le j) \\
           \wedge & (j' \le i \vee j \le j')
        \end{array} \right\}$ \\
       \putno\label{gen-union} \> \> \> $R_f$ := $R_f \cup \textsf{gen}(L, \langle w, S' \rangle, P')$ \\
       \putno \> \> \textbf{return} $R_f$
    }
  }
  \caption{\label{pcode:maxgen}
    Computing $\xi_{\overline{L}}(w)$, the maximum star-generalization of $w$ with respect to $\overline{L}$.
  }
\end{figure}
The procedure $\textsf{maxgen}$ carries around $S$, a partial generalization,
and $P$, the set of candidate $\star$-augmentations. At each stage, an
augmentation $e$ is selected from $P$, and we recursively
compute the set of valid further generalizations of $\langle w, S \rangle$
both including and excluding $e$,
finally taking the union of the sub-languages.

The maximum epsilon-generalization with respect to $L$ -- denoted by $\gen_{L}(w)$ --
may be constructed by an analogous procedure.

We now show that this generalization procedure is sufficiently powerful
as to prove separability for any pair of regularly separable languages.

\begin{lemma}\label{lemma-genword}
  Consider a context-free language $G$, and regular language $R$
  such that $G \cap R = \emptyset$.
  Then for any word $w \in R$, there is some
  $e' \in \kappa(R)$ such that $\lang{e'} \subseteq \xi_{\overline{G}}(w) $.
\end{lemma}
\begin{proof}
  By Lemma~\ref{lemma-reconstruct},
  there is some $e \in \Xi(w) \cap \kappa(R)$.
  As $R \cap G = \emptyset$, we have $e \in \Xi_{\overline{G}}(w)$.
  Therefore $\lang{e} \subseteq \bigcup \{\lang{e'} ~|~ e' \in \Xi_{\overline{G}}(w) \} = \xi_{\overline{G}}(w)$.
\end{proof}

\begin{corollary}\label{corrolary-epsword}
  Consider a context-free language $G$, and regular language $R$
  such that $G \cap R = \emptyset$.
  Then for any word $w \in R$, there is some
  $e' \in \kappa(R)$ such that $\lang{e'} \subseteq \mathsf{g}_{\overline{G}}(w) $.
\end{corollary}
\begin{proof}
  This follows immediately from Lemma~\ref{lemma-genword} and Theorem~\ref{thm-stareps}.
\end{proof}

\begin{theorem}\label{thm-complete}
  Given a pair of regularly separable context-free languages $(L, L')$
  and initial regular approximations $R_L$ and $R_{L'}$ with $L \subseteq R_L$
  and $L' \subseteq R_{L'}$, the refinement process described in
  Section~\ref{sec-refine}
will construct a separating pair $(S_L, S_{L'})$
  in a finite number of steps when refining using the maximum star- or
  epsilon-generalization.

\end{theorem}
\begin{proof}
  Consider the (unknown) regular language $S$ separating $L$ and $L'$.
  Assume $\kappa(S)$ and $\kappa(\overline{S})$ are as given in Theorem~\ref{lemma-reconstruct}.
  Let $K^i$ denote the elements of $\kappa(S) \cup \kappa(\overline{S})$ having non-empty
  intersection with the current approximation $R_{L}^i \cap R_{L'}^i$.

  Assume that, at a given step, there is some word $w \in R_L^i \cap R_{L'}^i$.
  $w$ must be in exactly one of $S$ and $\overline{S}$; we assume $w$ is in
  $S$ (the case of $\overline{S}$ is symmetric). As $w \in S$, there must
  be some $\recon \in \kappa(S)$ such that $\lang{\recon} \subseteq \xi_{\overline{L'}}(w) \subseteq \mathsf{g}_{\overline{L'}}(w)$.
  As $w \in R_L^i \cap R_{L'}^i$, and $w \in \recon$, we have $\recon \in K^i$.

  As $R_{L'}^{i+1} = R_{L'}^i \setminus \xi_{\overline{L'}}(w)$ (or $R_{L'}^i \setminus \mathsf{g}_{\overline{L'}}(w)$),
  we have
  $\lang{\recon} \cap R_{L'}^{i+1} = \emptyset$.
  Therefore, $K^{i+1} \subset K^{i}$. As $[K^1, K^2, \ldots]$ is a decreasing sequence,
  and $K^1$ is finite, the refinement process must terminate after finitely many steps.
\end{proof}

\section{Example}
\label{sec-example}
Consider the two context-free grammars $ G_1 = \langle \{S_1,A_1,B_1\},
\Sigma, P_1, S_1 \rangle $ and $ G_2 = \langle \{S_2,A_2,B_2\},
\Sigma, P_2, S_2 \rangle $ where $\Sigma = \{\ra, \rb\}$ and $P_1$ and
$P_2$ are, respectively,

\vspace{2mm}
\begin{tabular}{ll}
\begin{minipage}{0.5\textwidth}
\begin{tabular}{lcl}
$ S_1$ & $\rightarrow$ & $A_1 B_1 $ \\
$ A_1$ & $\rightarrow$ & $\ra\ra \mid \rb\rb \mid  \ra S_1 \ra \mid  \rb S_1 \rb  $ \\
$ B_1$ & $\rightarrow$ & $\ra\rb B_1 \mid \ra\rb  $ 
\end{tabular}
\end{minipage} &
\begin{minipage}{0.5\textwidth}
\begin{tabular}{lcl}
$ S_2$ & $\rightarrow$ & $A_2 B_2 $ \\
$ A_2$ & $\rightarrow$ & $\ra\ra \mid \rb\rb \mid  \ra S_2 \ra \mid  \rb S_2 \rb  $ \\
$ B_2$ & $\rightarrow$ & $ \rb\ra B_2 \mid \rb\ra  $ 
\end{tabular}
\end{minipage}
\end{tabular}
\vspace{2mm}

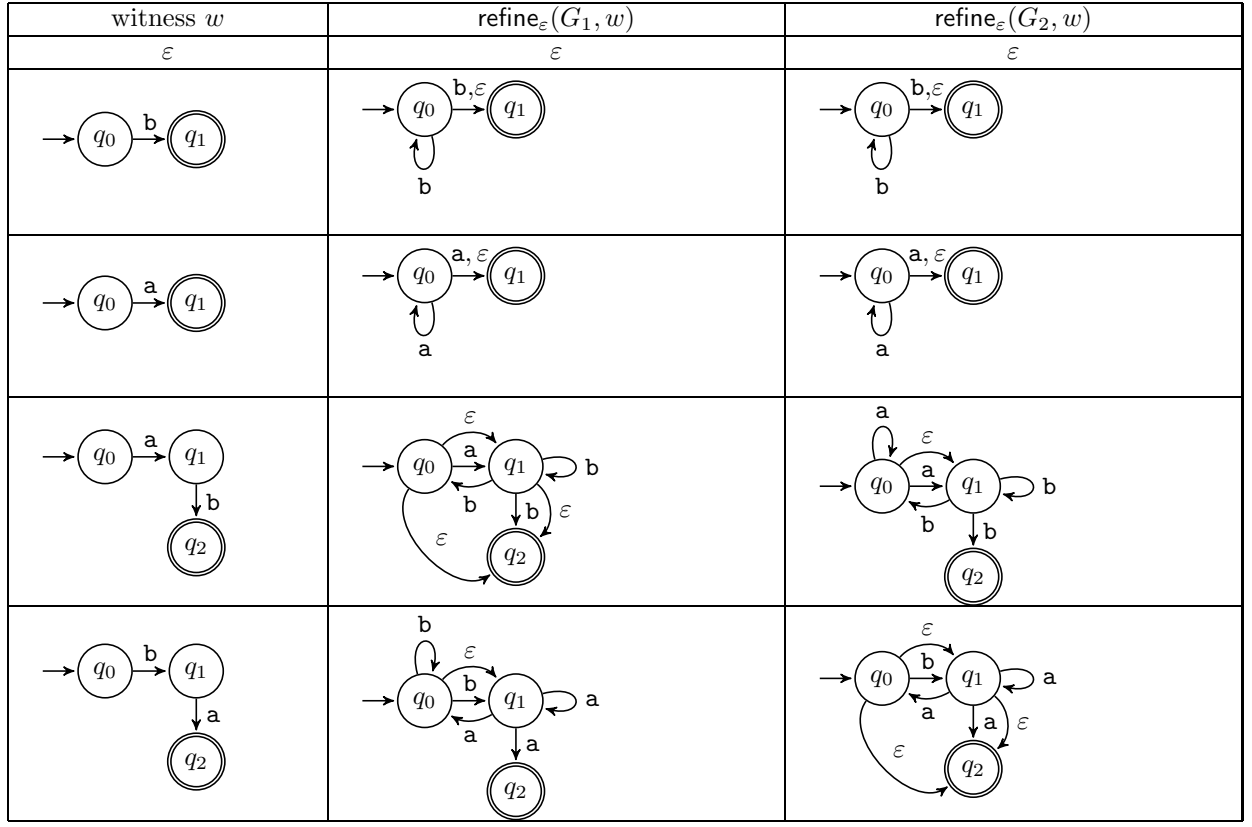
\begin{figure}[t]
\begin{tabular}{|c|c|c|}
\hline 
witness $w$  & \textsf{refine}$_{\eps}$($G_1, w $) & \textsf{refine}$_\eps$($G_2, w$) \\
\hline 
  \eps & \eps & \eps \\
\hline
\begin{minipage}{0.23\textwidth}
  \begin{tikzpicture}[->,>=stealth',shorten >=1pt,auto, node distance=1.2cm, semithick]
    \newState{q0}{$q_0$}{initial, initial text={}}{minimum size=0.5pt}
    \newState{q1}{$q_1$}{right of=q0}{accepting,minimum size=0.5pt}
    \newTransition{q0}{q1}{$\rb$}{}
  \end{tikzpicture}
\end{minipage} & 

\begin{minipage}{0.34\textwidth}
  \begin{tikzpicture}[->,>=stealth',shorten >=1pt,auto, node distance=1.2cm, semithick]
    \newState{q0}{$q_0$}{initial, initial text={}}{minimum size=0.5pt}
    \newState{q1}{$q_1$}{right of=q0}{accepting,minimum size=0.5pt}
    \newTransition{q0}{q1}{$\rb$,$\eps$}{}
    \newTransition{q0}{q0}{$\rb$}{loop below}
  \end{tikzpicture}
\end{minipage} & 

\begin{minipage}{0.34\textwidth}
  \begin{tikzpicture}[->,>=stealth',shorten >=1pt,auto, node distance=1.2cm, semithick]
    \newState{q0}{$q_0$}{initial, initial text={}}{minimum size=0.5pt}
    \newState{q1}{$q_1$}{right of=q0}{accepting,minimum size=0.5pt}
    \newTransition{q0}{q1}{$\rb$,$\eps$}{}
    \newTransition{q0}{q0}{$\rb$}{loop below}
  \end{tikzpicture} 
\end{minipage} 
\\ &&
\\ \hline 
\begin{minipage}{0.23\textwidth}
  \begin{tikzpicture}[->,>=stealth',shorten >=1pt,auto, node distance=1.2cm, semithick]
    \newState{q0}{$q_0$}{initial, initial text={}}{minimum size=0.5pt}
    \newState{q1}{$q_1$}{right of=q0}{accepting,minimum size=0.5pt}
    \newTransition{q0}{q1}{$\ra$\vphantom{$\rb$}}{}
  \end{tikzpicture}
\end{minipage} & 

\begin{minipage}{0.34\textwidth}
  \begin{tikzpicture}[->,>=stealth',shorten >=1pt,auto, node distance=1.2cm, semithick]
    \newState{q0}{$q_0$}{initial, initial text={}}{minimum size=0.5pt}
    \newState{q1}{$q_1$}{right of=q0}{accepting,minimum size=0.5pt}
    \newTransition{q0}{q1}{$\ra,\eps$\vphantom{$\rb,\eps$}}{}
    \newTransition{q0}{q0}{$\ra$}{loop below}
  \end{tikzpicture}
\end{minipage} & 

\begin{minipage}{0.34\textwidth}
  \begin{tikzpicture}[->,>=stealth',shorten >=1pt,auto, node distance=1.2cm, semithick]
    \newState{q0}{$q_0$}{initial, initial text={}}{minimum size=0.5pt}
    \newState{q1}{$q_1$}{right of=q0}{accepting,minimum size=0.5pt}
    \newTransition{q0}{q1}{$\ra,\eps$\vphantom{$\rb,\eps$}}{}
    \newTransition{q0}{q0}{$\ra$}{loop below}
  \end{tikzpicture} 
\end{minipage} 
\\ &&
\\ \hline 
\begin{minipage}{0.23\textwidth}
  \begin{tikzpicture}[->,>=stealth',shorten >=1pt,auto, node distance=1.2cm, semithick]
    \newState{q0}{$q_0$}{initial, initial text={}}{minimum size=0.5pt}
    \newState{q1}{$q_1$}{right of=q0}{minimum size=0.5pt}
    \newState{q2}{$q_2$}{below of=q1}{accepting,minimum size=0.5pt}
    \newTransition{q0}{q1}{$\ra$}{}
    \newTransition{q1}{q2}{$\rb$}{}
  \end{tikzpicture}
\end{minipage} & 

\begin{minipage}{0.34\textwidth}
  \begin{tikzpicture}[->,>=stealth',shorten >=1pt,auto, node distance=1.2cm, semithick]
    \newState{q0}{$q_0$}{initial, initial text={}}{minimum size=0.5pt}
    \newState{q1}{$q_1$}{right of=q0}{minimum size=0.5pt}
    \newState{q2}{$q_2$}{below of=q1}{accepting,minimum size=0.5pt}
    \newTransition{q0}{q1}{$\ra$}{}
    \newTransition{q1}{q2}{$\rb$}{}
    \newTransition{q1}{q1}{$\rb$}{loop right}
    \newTransition{q1}{q0}{$\rb$}{bend left}
    \newTransition{q0}{q1}{$\eps$}{bend left=50}
    \newTransition{q1}{q2}{$\eps$}{bend left=50}
    \newTransition{q0}{q2}{$\eps$}{bend right=80}
  \end{tikzpicture}
\end{minipage} & 

\begin{minipage}{0.34\textwidth}
  \begin{tikzpicture}[->,>=stealth',shorten >=1pt,auto, node distance=1.2cm, semithick]
    \newState{q0}{$q_0$}{initial, initial text={}}{minimum size=0.5pt}
    \newState{q1}{$q_1$}{right of=q0}{minimum size=0.5pt}
    \newState{q2}{$q_2$}{below of=q1}{accepting,minimum size=0.5pt}
    \newTransition{q0}{q1}{$\ra$}{}
    \newTransition{q1}{q2}{$\rb$}{}
    \newTransition{q0}{q0}{$\ra$}{loop above}
    \newTransition{q1}{q1}{$\rb$}{loop right}
    \newTransition{q1}{q0}{$\rb$}{bend left}
    \newTransition{q0}{q1}{$\eps$}{bend left=50}
  \end{tikzpicture}
\end{minipage} 
\\ \hline 
\begin{minipage}{0.23\textwidth}
  \begin{tikzpicture}[->,>=stealth',shorten >=1pt,auto, node distance=1.2cm, semithick]
    \newState{q0}{$q_0$}{initial, initial text={}}{minimum size=0.5pt}
    \newState{q1}{$q_1$}{right of=q0}{minimum size=0.5pt}
    \newState{q2}{$q_2$}{below of=q1}{accepting,minimum size=0.5pt}
    \newTransition{q0}{q1}{$\rb$}{}
    \newTransition{q1}{q2}{$\ra$}{}
  \end{tikzpicture}
\end{minipage} & 

\begin{minipage}{0.34\textwidth}
  \begin{tikzpicture}[->,>=stealth',shorten >=1pt,auto, node distance=1.2cm, semithick]
    \newState{q0}{$q_0$}{initial, initial text={}}{minimum size=0.5pt}
    \newState{q1}{$q_1$}{right of=q0}{minimum size=0.5pt}
    \newState{q2}{$q_2$}{below of=q1}{accepting,minimum size=0.5pt}
    \newTransition{q0}{q1}{$\rb$}{}
    \newTransition{q1}{q2}{$\ra$}{}
    \newTransition{q0}{q0}{$\rb$}{loop above}
    \newTransition{q1}{q1}{$\ra$}{loop right}
    \newTransition{q1}{q0}{$\ra$}{bend left}
    \newTransition{q0}{q1}{$\eps$}{bend left=50}
  \end{tikzpicture}
\end{minipage} & 

\begin{minipage}{0.34\textwidth}
  \begin{tikzpicture}[->,>=stealth',shorten >=1pt,auto, node distance=1.2cm, semithick]
    \newState{q0}{$q_0$}{initial, initial text={}}{minimum size=0.5pt}
    \newState{q1}{$q_1$}{right of=q0}{minimum size=0.5pt}
    \newState{q2}{$q_2$}{below of=q1}{accepting,minimum size=0.5pt}
    \newTransition{q0}{q1}{$\rb$}{}
    \newTransition{q1}{q2}{$\ra$}{}
    \newTransition{q1}{q1}{$\ra$}{loop right}
    \newTransition{q1}{q0}{$\ra$}{bend left}
    \newTransition{q0}{q1}{$\eps$}{bend left=50}
    \newTransition{q1}{q2}{$\eps$}{bend left=50}
    \newTransition{q0}{q2}{$\eps$}{bend right=80}
  \end{tikzpicture}
\end{minipage} 
\\ \hline  
\end{tabular} 
\caption{\label{fig:example-1}Relevant witnesses and generalizations obtained by
  greedy epsilon generalization.}
\end{figure}

\noindent 
Note that $\lang{G_1} = \{ ww^{R}(\ra\rb)^+ \mid w \in \Sigma^{*} \}$ 
and $\lang{G_2} = \{ ww^{R}(\rb\ra)^+ \mid w \in \Sigma^{*} \}$. 

The first step of our method will approximate $G_1$ and $G_2$ with
finite-state automata $A_1$ and $A_2$. The only requirement is that
$\lang{G_1} \subseteq \lang{A_1}$ and $\lang{G_2} \subseteq \lang{A_2}$. 
For simplicity, assume $A_1 = A_2 = \Sigma^{*}$.
Next, we check if $\lang{A_1} \cap \lang{A_2} \not= \emptyset$. 
In this case, the intersection is trivially not empty. 
Furthermore, our regular solver provides the witness $w=\eps$.
This cannot be generalized, so we eliminate $\eps$ from both approximations
and try again, this time obtaining $w=\rb$.
In the third step we refine the regular approximations. We assume the use
of greedy epsilon refinement, preferring backwards transitions.
We first generalize the witness by calling \textsf{refine}$_{\eps}$($G_1,w$) and
\textsf{refine}$_{\eps}$($G_2,w$) to produce new approximations $\lang{A_1'} = \lang{A_1}
\setminus \mathsf{refine}_{\eps}(G_1,w)$ and $\lang{A_2'} = \lang{A_2} \setminus
\mathsf{refine}_{\eps}(G_2,w)$, respectively.  We show the automata obtained
from $\mathsf{refine}_{\eps}(G_1,b)$ and $\mathsf{refine}_{\eps}(G_2,b)$ on the first row
in Figure~\ref{fig:example-1}. In both cases, we obtain the language $\rb^\star$.

Since $G_1$ and $G_2$ are regularly separable, using maximal refinement
would be guaranteed to eventually
eventually halt proving that the languages are disjoint.
While we have no such guarantee for greedy refinement, in this case
it successfully proves separation after 5 refinement steps.
Figure~\ref{fig:example-1} depicts the rest of witnesses obtained as
well as their generalizations produced by the procedure
\textsf{refine}$_{\eps}$.

\section{Previous Refinement Techniques}
\label{sec-compare}
Several CEGAR-based approaches have been proposed for testing
intersection of context-free languages. In this section, we attempt to
characterise the expressiveness of existing refinement methods.  For
these comparisons we do not consider the effect of initial regular
approximations, as they do not affect the expressiveness of the
refinement method.  For any fixed finite set of regularly-separable
languages, there is always some approximation scheme which allows the
languages to be trivially proven separate; however it is impossible to
define such an approximation in general.

The idea of using CEGAR to check the intersection of CFGs was
pioneered by Bouajjani \etal~\cite{BouajjaniET03} for the context of
verifying concurrent programs with recursive procedures.
They rely on the concept of \emph{refinable finite-chain abstraction}
consisting of computing the series $(\alpha_i)_{i \geq 1}$ which
overapproximates the language of a CFG $L$ (\ie, $L \subseteq
\alpha_i(L)$) such that $\alpha_1(L)\supset \alpha_2(L)\supset \cdots
\supseteq L$.  The method is parameterized by the refinable
abstraction. \cite{BouajjaniET03} describe several possible
abstractions but no experimental evaluation is provided.  Chaki
\etal~\cite{Chaki_TACAS06} extend~\cite{BouajjaniET03} by, among other
contributions, implementing and evaluating this method.
The experimental evaluation of Chaki \etal uses both the
\emph{$i^{th}$-prefix} and \emph{$i^{th}$-suffix} abstractions. Given
language $L$, the \emph{$i^{th}$-prefix} abstraction $\alpha_i(L)$ is
the set of words of $L$ of length less than $i$, together with the set
of prefixes of length $i$ of $L$. The \emph{$i^{th}$-suffix
  abstraction} can be defined analogously. 
We next provide a theorem about the 
expressiveness of the \emph{$i^{th}$-prefix abstraction}. 
A similar result holds for the \emph{$i^{th}$-suffix abstraction}.

\begin{theorem}
  There exist regularly separable languages that cannot be proven separate
  by the $i^{th}$-prefix abstraction.
\end{theorem}
\begin{proof}
  Consider the languages $R_1 = \ra^* \rb$, $R_2 = \ra^* \rc$. 
  $R_1 \cap R_2$ is empty. However,
  for a given length $i$, the string ${\ra}^i$ forms a prefix to
  words in both $R_1$ and $R_2$.
  It follows that the intersection of the two abstractions will always 
  be non-empty, so the refinement method cannot prove the languages
  separate.
\end{proof}

The \sys{lcegar} method described by Long \etal~\cite{LongCMM12}
is based on a similar refinement
framework, but the approach differs radically. They maintain
a pair of context-free grammars $A_1, A_2$ over-approximating the
intersection of the original languages.
At each refinement step, an \emph{elementary bounded language} $B_i$ is generated
from each grammar $A_i$.\footnote{
  An elementary bounded language is some language of the form
  $B = w_1^\star \ldots w_k^\star$, where each $w_i$ is a (finite) word in $\Sigma^\star$.
}
The refinement ensures $B_i \cap A_i \neq \emptyset$,
but $B_i$ is not necessarily either an over- or under-approximation of $A_i$.
They then compute $I = B_i \cap L_1 \cap L_2$. If $I$ is non-empty,
$L_1 \cap L_2$ must also be non-empty. If $I$ is empty, then the approximations
can safely be refined by subtracting the $B_i$.

We now wish to characterise the set of languages for which \sys{lcegar}
can prove separation.
Note that we do not consider the initial approximation; for any fixed
pair of regularly-separable languages, there necessarily exists \emph{some}
approximation method which immediately proves separation without refinement.
\begin{theorem}\label{thm-nonsep}
  There exist non-regularly-separable languages which can be proven separate
  by \sys{lcegar}.
\end{theorem}
\begin{proof}
  Consider the languages $L$ and $\overline{L}$,
  where $L = \{\ra^n \rb^n \mid n \geq 0\}$.
  These are not regularly separable.
  Still, \sys{lcegar} will find that they do not overlap.
  Assume initial approximations $A_1 = L_1$ and $A_2 = L_2$.
  At the first iteration, \sys{lcegar} may choose bounded approximation
  $B = \ra^* \rb^*$. 
  It will find $B \cap L_1 \cap L_2 = \emptyset$,
  then update $A_1 = A_1 \setminus B = \emptyset$.
  As $A_1 = \emptyset$, 
  the refinement process has successfully proven separation.
\end{proof}

\begin{lemma}\label{lemma-boundgap}
  For any bounded regular language $B = w_1^* \cdots w_k^*$, 
  there is some word $p$ that is not a substring of any word in $B$.
\end{lemma}
\begin{proof}
  For each word $w_{1}$, we pick some character $c_1$ which differs from the
  \emph{last} character of $w_{1}$.
  We then construct $p = p_1 \cdots p_k$, such that:
  \[ p_i = \underbrace{c_i \ldots c_i}_{\abs{w_i}} \]

  Assume there is some word $t = u p_1 \ldots p_k v \in B$. 
  $t$ must consist of some number of occurrences of $w_1$ through $w_k$, in order.
  Since $p_1$ differs from the last character of $w_1$, $u p_1$ cannot consist
  only of occurrences of $w_1$; therefore, $p_2 \cdots p_k$ must be made up of
  occurrences of $w_2$ through $w_k$.

  Similarly, since no occurrence of $w_2$ may end in $p_2$, so $p_3 \cdots p_k$
  must consist only of $w_3$ through $w_k$.
  By induction, we find that $p_k$ must be an occurrence of $w_k$. However, no
  occurrence of $w_k$ may occur in $p_k$.
  Therefore, there can be no word $t \in B$ such that $t \in \Sigma^* p \Sigma^*$.
\end{proof}

\begin{corollary}\label{cor-boundgap}
  For any finite set of bounded regular languages $\{B_1, \ldots, B_n\}$, we can
  construct some substring $p$ that is not a substring of $B_1 \cup \cdots \cup B_n$.
\end{corollary}
\begin{proof}
  By Lemma~\ref{lemma-boundgap}, we can find $p_1, \ldots, p_n$ such that
  $p_i$ is not a substring in $B_i$. 
  Then $p = p_1 \cdots p_n$ cannot occur as a
  substring in $B_1 \cup \ldots \cup B_n$.
\end{proof}

\begin{theorem}\label{thm-incomplete}
  There exist regularly separable languages for which the \sys{lcegar} 
  refinement method cannot prove separability.
\end{theorem}
\begin{proof}
  Consider an \sys{lcegar} process with $L_1 = A_1 = (\ra|\rb)^* \ra$,
  and $L_2 = A_2 = (\ra|\rb)^* \rb$.
  These languages are disjoint, and regularly separable.
  After some finite number of steps, the approximations
  have been refined with bounded languages $\{B_1, \ldots, B_n\}$.
  By Corollary~\ref{cor-boundgap}, there is some substring $p$ such that
  $\Sigma^* p\Sigma^* \subseteq (\overline{B_1} \cap \ldots \cap \overline{B_n})$.
  The updated approximation $A'_1$ is non-empty, as it contains $p\ra$.
  Similarly, the approximation $A'_2$ is non-empty, as it contains $p\rb$.

  Since after any finite sequence of refinement steps neither $A'_1$ nor
  $A'_2$ is empty, the refinement process will never prove separation
  of $L_1$ and $L_2$.
\end{proof}

From Theorems~\ref{thm-complete},~\ref{thm-nonsep} and~\ref{thm-incomplete},
we conclude that the classes of languages which can be proven separate by
\sys{lcegar} and \oursolver{} are incomparable.

\section{Experimental Evaluation}
\label{sec-results}

We have implemented the CEGAR method proposed in this paper in a
prototype tool called \oursolver\footnote{Publicly available at
  \url{https://bitbucket.org/jorgenavas/covenant} together with all
  the benchmarks used in this section.}. The tool is implemented in
C++ and parameterized by the initial approximation and the refinement
procedure. \oursolver\ implements the method described
in~\cite{Nederhof_chapter1} for approximating CFGs with strongly
regular languages as well as the coarsest abstraction $\Sigma^{*}$ for
comparison purposes. For refinement, the tool implements both the
greedy and maximum star-epsilon generalizations (described in
Sections~\ref{sec-stargen} and~\ref{sec-maxgen},
respectively). \oursolver\ currently implements only the classical
product construction for solving the intersection of regular languages
but other regular solvers (\eg, \cite{revenant,strsolve}) can be
easily integrated\footnote{In fact, an initial implementation of
  \oursolver\ was tested using \sys{Revenant}~\cite{revenant}, an
  efficient regular solver based on bounded model checking with
  interpolation, though the released version does not incorporate
  it.}.

To assess the effectiveness of our tool, we have conducted
two experiments. First, we used \oursolver\ for proving safety
properties in recursive multi-threaded programs. Second, we crafted
pairs of challenging context free grammars and intersected them using
\oursolver. The motivation for this second experiment was to exercise
features of \oursolver\ that were not required during the first
experiment.
All experiments were run on a single core of a 2.4GHz Core i5-M520
with 7.8Gb memory.

\vspace{2mm}
\noindent
\textbf{Safety verification of recursive multi-threaded programs.}
Bouajjani \etal~\cite{BouajjaniET03} was pioneered showing that the
safety verification problem of recursive multi-threaded programs can
be reduced to check whether the intersection of context free languages
is empty. Since then, several encodings have been
described~\cite{BouajjaniET03,Chaki_TACAS06,LongCMM12}.
As a result, we can use \oursolver\ to prove certain safety properties
in recursive multi-threaded programs assuming the programs have been
translated accordingly.
We briefly exemplify the translation of a concurrent program to
context-free grammars, using the approach of~\cite{LongCMM12}.

For simplicity, we assume a concurrency model in which 
communication is based on shared memory. Shared memory is modelled via
a set of global variables.  We assume that each statement is executed
atomically.
We will consider only \emph{Boolean programs}. Any program $P$ can be
translated into a Boolean program $\mathcal{B}(P)$ using techniques
such as predicate abstraction~\cite{GrafS97}. A key property is that
$\mathcal{B}(P)$ is an over-approximation of $P$ preserving the
control flow of $P$ but the only type available in $\mathcal{B}(P)$ is
Boolean. Therefore, if $\mathcal{B}(P)$ is correct then $P$ must be
correct but, of course, if $\mathcal{B}(P)$ is unsafe $P$ may be still
safe.
Each (possibly recursive) procedure in $\mathcal{B}(P)$ is modelled as
a context-free grammar as well as each shared variable specifying the
possible values that the variable can take. In addition, extra
production rules are added to specify the synchronization points.

\newcommand*{\myhline}{\rule{\textwidth}{0.55pt}}

\begin{figure}[t]
\begin{tabular}[t]{|l|l|l|}
\hline
\ppcode{
$x=0; y=0;$ \\[2mm]
p1 () $\{$ \\
$n_{0}$: \> \> $x = \mathsf{not}~ y$ ;   \\
$n_{1}$: \> \> \texttt{if}($*$) p1();  \\
$n_{2}$: \> \> $x = \mathsf{not}~ y$ ;\\
$n_{3}$: $\}$ \\[2mm]
p2 () $\{$ \\
$m_{0}$: \> \> $y = \mathsf{not}~ x$ ;   \\
$m_{1}$: \> \> \texttt{if}($*$)  p2();  \\
$m_{2}$: \> \> $y = \mathsf{not}~ x$ ; \\
$m_{3}$: $\}$ \\[2mm] 

\texttt{if} $(x ~\mathsf{and}~ y)$ \\
\> \textsf{error}(); \\ 
} & 
\begin{minipage}{0.43\textwidth}
\begin{tabular}{@{}p{0.2\linewidth}p{0.01\linewidth}p{0.55\linewidth}@{}}
\multicolumn{3}{@{}c}{\scriptsize$\mathsf{CFG_1}$} \\[-2mm]
\myhline \\[-2mm]
\multicolumn{3}{@{}l}{\scriptsize// control flow of thread $p1$} \\
\nt{N}{0} & $\rightarrow$ & $\encode{x = \mathsf{not}~ y}$ \nt{N}{1} \\
\nt{N}{1} & $\rightarrow$ &  \nt{N}{0} \nt{N}{2} $\mid$ \nt{N}{2} \\
\nt{N}{2} & $\rightarrow$ & $\encode{x = \mathsf{not}~ y}$ \nt{N}{3} \\
\nt{N}{3} & $\rightarrow$ & $\encode{x}$ \\
\multicolumn{3}{@{}l}{\scriptsize//encoding of instructions for $p1$} \\
$\encode{x = \mathsf{not}~ y}$ & $\rightarrow$ & \nt{S}{p_{2}}  \readterm{y}{0} \writeterm{x}{1} \nt{S}{p_{2}} $\mid$ \\
 &  & \nt{S}{p_{2}}  \readterm{y}{1} \writeterm{x}{0} \nt{S}{p_{2}} \\
$\encode{x}$ & $\rightarrow$ & \readterm{x}{1}  \\
\multicolumn{3}{@{}l}{\scriptsize//synchronization with $p2$'s actions} \\
\nt{S}{p_{2}} & $\rightarrow$ & \readterm{x}{0} \nt{S}{p_{2}} $\mid$ \\
& &   \readterm{x}{1} \nt{S}{p_{2}} $\mid$ \\
& & \writeterm{y}{0} \nt{S}{p_{2}} $\mid$ \\
& & \writeterm{y}{1} \nt{S}{p_{2}} $\mid$ $\epsilon$ \\
\myhline \\[-1mm]
\multicolumn{3}{@{}c}{\scriptsize$\mathsf{CFG_2}$} \\[-2mm]
\myhline \\[-2mm]
\multicolumn{3}{@{}l}{\scriptsize//Control flow of thread $p2$} \\
\nt{M}{0} & $\rightarrow$ & $\encode{y = \mathsf{not}~ x}$ \nt{M}{1} \\
\nt{M}{1} & $\rightarrow$ &  \nt{M}{0} \nt{M}{2} $\mid$ \nt{M}{2}\\
\nt{M}{2} & $\rightarrow$ & $\encode{y = \mathsf{not}~ x}$ \nt{M}{3} \\
\nt{M}{3} & $\rightarrow$ & $\encode{y}$ \\
\multicolumn{3}{@{}l}{\scriptsize//Encoding of instructions for $p2$} \\
$\encode{y = \mathsf{not}~ x}$ & $\rightarrow$ & \nt{S}{p_{1}}  \readterm{x}{0} \writeterm{y}{1} \nt{S}{p_{1}} $\mid$ \\
 &  & \nt{S}{p_{1}}  \readterm{x}{1} \writeterm{y}{0} \nt{S}{p_{1}} \\
$\encode{y}$ & $\rightarrow$ & \readterm{y}{1}  \\
\multicolumn{3}{@{}l}{\scriptsize//Synchronization with $p1$'s actions} \\
\nt{S}{p_{1}} & $\rightarrow$ & \readterm{y}{0} \nt{S}{p_{1}} $\mid$ \\
& & \readterm{y}{1} \nt{S}{p_{1}} $\mid$ \\
& & \writeterm{x}{0} \nt{S}{p_{1}} $\mid$ \\
& & \writeterm{x}{1} \nt{S}{p_{1}} $\mid$ $\epsilon$ \\
\end{tabular}
\end{minipage}
& 
\begin{minipage}{0.3\textwidth}
\begin{tabular}{@{}p{0.1\linewidth} p{0.01\linewidth} p{0.5\linewidth}@{}}
\multicolumn{3}{@{}c}{\scriptsize$\mathsf{CFG_3}$} \\[-2mm]
\myhline \\[-2mm]
\multicolumn{3}{@{}l}{\scriptsize//Modelling variable $x$} \\
\nt{X}{false} & $\rightarrow$ & \readterm{x}{0} \nt{X}{false} $\mid$ \\
& & \writeterm{x}{0} \nt{X}{false} $\mid$ \\
& & \writeterm{x}{1} \nt{X}{true} $\mid$ \\
& & \nt{S}{x} \nt{X}{true} $\mid$ $\epsilon$   \\
\nt{X}{true} & $\rightarrow$ & \readterm{x}{1} \nt{X}{true} $\mid$ \writeterm{x}{1} \nt{X}{true} $\mid$\\
& & \writeterm{x}{0} \nt{X}{false} $\mid$ \nt{S}{x} \nt{X}{true} $\mid$ $\epsilon$   \\
\multicolumn{3}{@{}l}{\scriptsize//Synchronization with $y$} \\
\nt{S}{x} & $\rightarrow$ & \readterm{y}{0} \nt{S}{x} $\mid$ \\
& &  \writeterm{y}{0} \nt{S}{x} $\mid$ \\
& & \readterm{y}{1} \nt{S}{x} $\mid$ \\
& & \writeterm{y}{1} \nt{S}{x} $\mid$ $\epsilon$ \\
\myhline \\[-1mm]
\multicolumn{3}{@{}c}{\scriptsize$\mathsf{CFG_4}$} \\[-2mm]
\myhline \\[-2mm]
\multicolumn{3}{@{}l}{\scriptsize//Modelling variable $y$} \\
\nt{Y}{false} & $\rightarrow$ & \readterm{y}{0} \nt{Y}{false} $\mid$ \\
& & \writeterm{y}{0} \nt{Y}{false} $\mid$ \\
& & \writeterm{y}{1} \nt{Y}{true} $\mid$ \\
& & \nt{S}{y} \nt{Y}{true} $\mid$ $\epsilon$   \\
\nt{Y}{true} & $\rightarrow$ & \readterm{y}{1} \nt{Y}{true} $\mid$ \writeterm{y}{1} \nt{Y}{true} $\mid$\\
& & \writeterm{y}{0} \nt{Y}{false} $\mid$ \nt{S}{y} \nt{Y}{true} $\mid$ $\epsilon$   \\
\multicolumn{3}{@{}l}{\scriptsize//Synchronization with $x$} \\
\nt{S}{y} & $\rightarrow$ & \readterm{x}{0} \nt{S}{y} $\mid$ \\
& & \writeterm{x}{0} \nt{S}{y} $\mid$ \\
& & \readterm{x}{1} \nt{S}{y} $\mid$ \\
& & \writeterm{x}{1} \nt{S}{y} $\mid$ $\epsilon$ \\
\end{tabular}
\end{minipage}
\\
\hline
\end{tabular}
\caption{\label{fig:encoding} A concurrent Boolean program
  (\textsf{SharedMem}) and its translation to CFGs.
  }
\end{figure}

The left hand column of Figure~\ref{fig:encoding} shows a small program
\textsf{SharedMem}~\cite{LongCMM12}. It consists of two symmetric,
recursive procedures $p1$ and $p2$ which are executed by two different
threads. The communication between the threads is done through the
global variables $x$ and $y$ which are initially set to $0$. Note that
the program is already Boolean since $x$ and $y$ can only take values
$0$ and $1$. We would like to prove that after $p1$ and $p2$
terminate, $x$ and $y$ cannot be true simultaneously.
 
The rest of Figure~\ref{fig:encoding} describes the
corresponding translation to context-free grammars.  
The four resulting grammars, which we explain shortly, are
\[
\begin{array}{ll}
   \mathsf{CFG}_1 : &
     \tuple{\{ \nt{N}{0}, \nt{N}{1}, \nt{N}{2}, \nt{N}{3}, 
	\encode{x = \mathsf{not}~ y}, \encode{x}, \nt{S}{p_{2}} \},
	\Sigma,P_1, \nt{N}{0}}
\\ \mathsf{CFG}_2 : &
     \tuple{\{ \nt{M}{0}, \nt{M}{1}, \nt{M}{2}, \nt{M}{3}, 
	\encode{y = \mathsf{not}~ x}, \encode{y}, \nt{S}{p_{1}} \},
	\Sigma,P_2, \nt{M}{0}}
\\ \mathsf{CFG}_3 : &
     \tuple{\{ \nt{X}{false}, \nt{X}{true}, \nt{S}{x}\},
	\Sigma,P_3, \nt{X}{false}}
\\ \mathsf{CFG}_4 : &
     \tuple{\{ \nt{Y}{false}, \nt{Y}{true}, \nt{S}{y}\},
	\Sigma,P_4, \nt{Y}{false}}
\end{array}
\]
where $\Sigma= \{ \readterm{x}{0}, \readterm{x}{1},
  \readterm{y}{0}, \readterm{y}{1}, \writeterm{x}{0},
  \writeterm{x}{1}, \writeterm{y}{0}, \writeterm{y}{1}\}$
and $P_1$, $P_2$, $P_3$, and $P_4$ are the respective sets of
productions, as shown in Figure~\ref{fig:encoding}.

Procedures $p1$ and $p2$ are translated into $\mathsf{CFG_1}$ and 
$\mathsf{CFG_2}$, respectively. 
First, we need to encode the control flow of the
procedures.
For instance, ``$p1$ reaches location $n_0$ and it executes the
statement $x =~\mathsf{not}~y$'' is translated into the grammar
production $\nt{N}{0} \rightarrow \encode{x = \mathsf{not}~
  y}~\nt{N}{1}$, where $\nt{N}{1}$ represents the next program
location $n_1$. We use the notation $\encode{s} \in V$ to refer to the
corresponding translation of statement $s$.
A function call such as ``$p1$ calls itself recursively after location
$n_1$ is executed'' is translated through the production $\nt{N}{1}
\rightarrow \nt{N}{0} \nt{N}{2}$ where $\nt{N}{0}$ is the entry
location of the callee function and $\nt{N}{2}$ is the continuation of
the caller after the callee returns.
The non-terminal symbol $\encode{x = \mathsf{not}~y}$ models the
execution of negating $y$ and storing its result in $x$. We create a
terminal symbol for each possible action on $x$ (and analogously for
$y$): \readterm{x}{0} (the value of $x$ is $0$), \readterm{x}{1} (the
value of $x$ is $1$), \writeterm{x}{0} ($x$ is updated to~$0$), and
\writeterm{x}{1} ($x$ is updated to~$1$). For instance, the grammar
production $\encode{x = \mathsf{not}~ y}~\rightarrow ~\nt{S}{p_{2}}~
\readterm{y}{0}~\writeterm{x}{1}~ \nt{S}{p_{2}}$ represents that if we
read $0$ as the value of $y$ then it must be followed by writing $1$
to $x$. The rest of logical operations are encoded similarly.

Note that whenever a global variable is read or written we need to
consider the synchronization between threads. For this purpose, we
define the non-terminal symbols $\nt{S}{p_{2}}$ ($\nt{S}{p_{1}}$)
which loops zero or more times with all possible actions of $p2$
($p1$): value of $x$ is $0$ (value of $y$ is $0$), value of $x$
is $1$ (value of $y$ is $1$), $y$ is updated to $0$ ($x$ is updated to
$0$), and $y$ is updated to $1$ ($x$ is updated to $1$).

Next, we need to model which are the possible values that $x$ and $y$
can take. For this we use $\mathsf{CFG_3}$ and $\mathsf{CFG_4}$, 
respectively. Ignoring synchronization, the set of
values that $x$ and $y$ can take are indeed expressed by regular
automata:

\vspace{2mm}
\begin{tabular}{cc}
\begin{minipage}{0.5\textwidth}
  \begin{tikzpicture}[->,>=stealth',shorten >=1pt,auto, 
      node distance=4cm, semithick]
    \newState{q0}{$\nt{X}{false}$}{initial, initial text={}}{accepting,minimum size=0.5pt}
    \newState{q1}{$\nt{X}{true}$}{right of=q0}{accepting, minimum size=0.5pt}
    \newTransition{q0}{q0}{$\readterm{x}{0}$}{loop above}
    \newTransition{q0}{q0}{$\writeterm{x}{0}$}{loop below}
    \newTransition{q0}{q1}{$\writeterm{x}{1}$}{bend left=30}
    \newTransition{q1}{q1}{$\readterm{x}{1}$}{loop above}
    \newTransition{q1}{q1}{$\writeterm{x}{1}$}{loop below}
    \newTransition{q1}{q0}{$\writeterm{x}{0}$}{bend right=-30}

  \end{tikzpicture}
\end{minipage} &
\begin{minipage}{0.5\textwidth}
  \begin{tikzpicture}[->,>=stealth',shorten >=1pt,auto, 
      node distance=4cm, semithick]
    \newState{q0}{$\nt{Y}{false}$}{initial, initial text={}}{accepting,minimum size=0.5pt}
    \newState{q1}{$\nt{Y}{true}$}{right of=q0}{accepting, minimum size=0.5pt}
    \newTransition{q0}{q0}{$\readterm{y}{0}$}{loop above}
    \newTransition{q0}{q0}{$\writeterm{y}{0}$}{loop below}
    \newTransition{q0}{q1}{$\writeterm{y}{1}$}{bend left=30}
    \newTransition{q1}{q1}{$\readterm{y}{1}$}{loop above}
    \newTransition{q1}{q1}{$\writeterm{y}{1}$}{loop below}
    \newTransition{q1}{q0}{$\writeterm{y}{0}$}{bend right=-30}
  \end{tikzpicture}
\end{minipage} 
\end{tabular}
\vspace{2mm}

Finally, we need to synchronize $x$ and $y$ by allowing them to loop
zero or more times while new values from the other variable can be
generated. We use non-terminal symbols (and their productions)
$\nt{S}{x}$ and $\nt{S}{y}$ for that.

\begin{table}[p]
  \begin{center}
    \subfloat[Verification of multi-thread Erlang programs]{
      \setlength{\tabcolsep}{16pt}
    \begin{tabular}{|l|c||c||c|c|}
      \hline
      \multicolumn{2}{|c||}{Program} & \oursolver & \multicolumn{2}{|c|}{\sys{lcegar}} \\
      \hline
      \multicolumn{2}{|c||}{}  & & PDC & CB \\
      \hline
      \hline
      \textsf{SharedMem} & \hphantom{un}safe & 0.01 & 14.37 & 24.75 \\
      \hline
      \textsf{Mutex} & \hphantom{un}safe & 0.04 & \hphantom{1}6.12  & \hphantom{2}0.14 \\
      \hline
      \textsf{RA} & \hphantom{un}safe & 0.01 & $\infty$  & \hphantom{2}0.39 \\
      \hline
      \textsf{Modified RA} & \hphantom{un}safe & 0.03 & $\infty$  & 27.90 \\
      \hline
      \textsf{TNA} & unsafe &  0.01 & \hphantom{1}0.02 & \hphantom{2}0.25 \\
      \hline
      \textsf{Banking} & unsafe & 0.01 & $\infty$  & \hphantom{2}3.36  \\
      \hline
    \end{tabular}
    }\\
    \subfloat[Verification of multi-thread Bluetooth drivers]{
      \setlength{\tabcolsep}{9pt}
    \begin{tabular}{|l|c||c||r|r|}
      \hline
      \multicolumn{2}{|c||}{Program} & \oursolver & \multicolumn{2}{|c|}{\sys{lcegar}} \\       
      \hline
      \multicolumn{2}{|c||}{}  & & PDC & CB \\
      \hline
      \hline
      \textsf{Version 1} & unsafe & \hphantom{2}0.84 & 19.74 & 21.04 \\
      \hline
      \textsf{Version 2} & unsafe & \hphantom{2}0.25  & 5560.00 & 4852.00  \\
      \hline
      \textsf{Version 2 w/ Heuri} & unsafe & \hphantom{2}0.11  & 44.68  & 38.89 \\
      \hline
      \textsf{Version 3 (1A2S)} & unsafe & \hphantom{2}0.12 & 217.74  & 217.27 \\
      \hline
      \textsf{Version 3 (1A2S) w/ Heuri} & unsafe & \hphantom{2}0.05 & 6.68   & 11.37 \\
      \hline
      \textsf{Version 3 (2A1S)} & \hphantom{un}safe & \hphantom{2}0.27 & 4185.00  & 3981.00  \\
      \hline
    \end{tabular}
    }\\
    \subfloat[Interesting/challenging grammars ($\infty$ indicates time-out at 60 sec and ``-'' a raised exception.)]{
      \setlength{\tabcolsep}{4pt}
    \begin{tabular}{|c|c||c|c|c|c||c|c|}
      \hline
      \multicolumn{2}{|c||}{} & 
      \multicolumn{4}{c||}{\oursolver} & 
      \multicolumn{2}{c|}{\sys{lcegar}} \\
      \cline{3-8}
      \multicolumn{2}{|c||}{} &
      \multicolumn{2}{|c|}{$\Sigma^{*}$} &
      \multicolumn{2}{|c||}{\cite{Nederhof_chapter1}} &  
      \multicolumn{1}{|c|}{PDC}  &
      \multicolumn{1}{|c|}{CB}   \\
      \cline{3-8}
      \multicolumn{2}{|c||}{} & Greedy & Gen & Greedy & Gen &  &  \\
      \hline 
      \hline 
      $C_1 \cap C_7$ & \hphantom{un}sat   
                     & \hphantom{2}8 (0.01)   
                     & \hphantom{1}11 (7.88) 
                     & \hphantom{2}5 (0.01)   
                     & \hphantom{2}8 (6.20)   
                     & $\infty$    &   -- \\  
      $C_7 \cap C_1$ &       &            &          &          &          & \hphantom{1}0 (0.13) & 0 (0.32) \\
      \hline
      $C_1 \cap C_8$ & \hphantom{un}sat   
                     & \hphantom{2}8 (0.01)   
                     & \hphantom{1}13 (8.36) 
                     & \hphantom{2}7 (0.01)   
                     & \hphantom{2}9 (2.22)   
                     & \hphantom{2}0 (20.28) & -- \\ 
      $C_8 \cap C_1$ &       &            &          &          &          & $\infty$ &  $\infty$  \\
      \hline
      $C_2 \cap C_3$ & \hphantom{un}sat   
                     & \hphantom{1}10 (0.01)  
                     & \hphantom{1}13 (9.10) 
                     & \hphantom{2}2 (0.01)   
                     & \hphantom{2}2 (0.02)   
                     & 0 (0.03) & 0 (0.01) \\  
      $C_3 \cap C_2$ &       &            &           &          &          & 0 (0.03) & 0 (0.01) \\
      \hline
      $C_2 \cap C_4$ & unsat 
                    & \hphantom{1}15 (0.02) 
                    & $\infty$              
                    &  \hphantom{2}3 (0.01) 
                    & \hphantom{2}3 (0.80)  
                    &  1 (0.01) & 0 (0.01) \\ 
      $C_4 \cap C_2$ &       &    &          &           &          &  $\infty$ & 0 (0.01) \\
      \hline
      $C_3 \cap C_4$ & unsat 
                     & \hphantom{1}11 (0.01) 
                     & $\infty$              
                     & \hphantom{2}2 (0.01)  
                     & \hphantom{2}2 (0.04)  
                     & 0 (0.01) & 0 (0.01) \\ 
      $C_4 \cap C_3$ &       &    &          &          &          & 0 (0.01) & 0 (0.01) \\
      \hline
      $C_5 \cap C_6$ & unsat 
                     & \hphantom{2}6 (0.01) 
                     & $\infty$             
                     & \hphantom{2}5 (0.01) 
                     & $\infty$             
                     & $\infty$ & 0 (0.01)  \\ 
      $C_6 \cap C_5$ &       &          &          &    &          & $\infty$ & 0 (0.01) \\
      \hline
      $C_5 \cap C_7$ & \hphantom{un}sat 
                     & \hphantom{1}14 (0.04) 
                     & $\infty$              
                     & \hphantom{1}11 (0.02) 
                     & $\infty$              
                     &  $\infty$ & -- \\ 
      $C_7 \cap C_5$ &       &         &     &           &    &  0 (0.33) & $\infty$  \\
      \hline
      $C_5 \cap C_8$ & \hphantom{un}sat 
                     & \hphantom{2}7 (0.01) 
                     & \hphantom{2}9 (2.81) 
                     & \hphantom{2}5 (0.01) 
                     & \hphantom{2}5 (3.54) 
                     & $\infty$ & -- \\ 
      $C_8 \cap C_5$ &       &        &          &          &          & 0 (0.04) & $\infty$ \\
      \hline
      $C_6 \cap C_7$ & \hphantom{un}sat  
                     & \hphantom{1}14 (0.04) 
                     & $\infty$              
                     & \hphantom{1}11 (0.02) 
                     & $\infty$              
                     & $\infty$  & -- \\  
      $C_7 \cap C_6$ &      &          &    &          &    &  0 (0.10) & $\infty$ \\
      \hline
      $C_6 \cap C_8$ & \hphantom{un}sat 
                     & \hphantom{2}8 (0.01) 
                     & \hphantom{2}9 (2.86) 
                     & \hphantom{2}5 (0.01) 
                     & \hphantom{2}5 (3.46) 
                     & 0 (1.21) & -- \\ 
      $C_8 \cap C_6$ &       &        &          &          &          & $\infty$ & $\infty$ \\
      \hline
      $C_7 \cap C_8$ & \hphantom{un}sat 
                     & \hphantom{2}4 (0.01) 
                     & \hphantom{2}4 (0.01) 
                     & \hphantom{2}3 (0.01) 
                     & \hphantom{2}3 (0.01) 
                     & 0 (0.70) & --  \\  
      $C_8 \cap C_7$ &     &          &          &          &          & $\infty$ & -- \\
      \hline
    \end{tabular}
    }
  \end{center}
  \caption{\label{results} Comparison of \oursolver\ with
    \sys{lcegar}, on several classes of context free grammars;
    times in seconds.  }
\end{table}

Once we have obtained the CFGs described in Figure~\ref{fig:encoding}
we are ready to ask reachability questions. For this example, we would
like to prove that when threads start at $n_{0}$ and $m_{0}$,
respectively, \textsf{error} cannot be reachable simultaneously by both
threads. This question can be answered by checking if the intersection
of the above CFGs is empty. If the intersection is not empty then
\oursolver\ will return a witness $w \in \Sigma^{*}$ containing the
sequence of reads and writes to $x$ and $y$. Otherwise, 
\oursolver\ will return either ``yes'' (that is, the program is safe) if
the languages of the CFGs are regularly separable or run until
resources are exhausted.

We have tested \oursolver\ with the programs used in~\cite{LongCMM12}
and compared with~\sys{lcegar}~\cite{LongCMM12}. There are two classes
of programs: textbook Erlang programs and several variants of a real
Bluetooth driver. The Bluetooth variants labelled \textsf{W/
  Heuri} are encoded with an unsound heuristic that permits context
switches only at basic block boundaries.
We refer readers to~\ref{sec-benchmarks} for a detailed
description of the programs as well as the safety properties.

Table~\ref{results}(a) and Table~\ref{results}(b) show the times in
seconds for both solvers when proving the Erlang programs and the
Bluetooth drivers. The symbol $\infty$ indicates that the solver
failed to terminate after $2$ hours.
We ran~\sys{lcegar}
using the settings suggested by the authors and tried with the two
available initial abstractions: \emph{pseudo-downward closure} (PDC)
and \emph{cycle breaking} (CB).
For our tool, we used as the initial abstraction the one described
in~\cite{Nederhof_chapter1} which is described in~\ref{sec-abstract}.
We also tried $\Sigma^{*}$ but \oursolver\ did not converge for any of
the programs in a reasonable amount of time.

It is somewhat surprising that all properties were successfully
proven by \sys{lcegar} using the initial regular approximation,
including Bluetooth instances.  The same is true for \oursolver,
except for \textsf{Version 1} which required $12$ refinements using
the greedy strategy.  Nevertheless, these programs show cases in which
\oursolver\ can significantly outperform \sys{lcegar}. Since almost no
refinements were required by any of the tools, it also suggests that
the approximation of all CFGs at once and the use of a regular solver
is often a more efficient choice than relying on computing
intersection of CFLs and regular languages as \sys{lcegar} does.

\vspace{2mm}
\noindent \textbf{Other interesting CFLs.}  The
verification instances \cite{LongCMM12} are in fact all solved with no
use of refinement, by \sys{lcegar} as well as by \oursolver\ (with the
exception of one instance).  To explore more interesting cases that
exercise the refinement procedures, we have added experiments
involving the following languages ($\Sigma = \{\ra,\rb\}^*$; note that
$C_5$ is $\lang{G_1}$ from Section~\ref{sec-example} and $C_6$ is 
$\lang{G_2}$):

\begin{center}
\setlength{\tabcolsep}{8pt}
\begin{tabular}{l|l}
 $C_1:$  $\{w w^R \mid w \in \Sigma^* \}$   & $C_5:$  $\{w w^R (\ra \rb)^+ \mid w \in \Sigma^* \}$ \\ 
 $C_2:$  $\{w c w^R \mid w \in \Sigma^* \}$ & $C_6:$  $\{w w^R (\rb \ra)^+ \mid w \in \Sigma^* \}$ \\
 $C_3:$  $\{{\ra}^n \rc {\ra}^n \mid n > 0 \}$      & $C_7:$  $\{w \in \Sigma^* \mid \mbox{$w$ has equal numbers of $\ra$s and $\rb$s} \}$ \\
 $C_4:$  $\{{\ra}^n \rc {\rb}^n \mid n > 0 \}$  & $C_8:$  $\{w w' \mid |w| = |w'|, w \not= w' \}$   \\
\end{tabular}
\end{center}

\noindent
Table~\ref{results}(c) shows the pairs of languages whose disjointness
can be proven or a counterexample can be found requiring at least one
refinement for \oursolver. We ignore pairs of languages which
are disjoint but not regularly separable.

We ran \oursolver\ using two initial abstractions: $\Sigma^{*}$ and
the more precise one described by Nederhof~\cite{Nederhof_chapter1}.
For each one, we used our
greedy (Greedy) refinement described in Section~\ref{sec-stargen} and
the complete refinement (Gen) from Section~\ref{sec-maxgen}. We
compared again with \sys{lcegar} using its two abstractions PDC and
CB. For a given $C_i \cap C_j$ \sys{lcegar} checks first whether
$\lang{C_i} \cap \lang{\alpha(C_j)} = \emptyset$ and then, only if it
is not empty a refinement is triggered. Therefore, \sys{lcegar} fixes
a priori that the first input grammar will not be abstracted while the
second will (that is, the order in which the grammars are given to
\sys{lcegar} matters). 
That is why we test both $C_i \cap C_j$ and $C_j \cap C_i$. 
In \oursolver\ the order is irrelevant.
We use the format \#R (T) to indicate that the tool needed \#R
refinements to prove disjointness or to find a counterexample, in T
seconds. We set a timeout ($\infty$) of 60 seconds.

Table~\ref{results}(c) indicates that, generally, the more precise the
initial abstraction, the fewer refinements are necessary.  This claim
was also made in~\cite{LongCMM12} although we were not able to fully
confirm it because \sys{lcegar} raised an exception with many of the
instances while using CB (denoted by the symbol --).  Interestingly,
Greedy performs quite well, terminating for all instances. This
suggests that Greedy might be a good practical choice in cases where
Gen spends too much time computing the generalization of the
witnesses.

Regarding \sys{lcegar} either a timeout is reached or the tool can
either prove disjointness or find a witness without any refinement
except for one instance ($C_2 \cap C_4$). It is worth noticing that
even if both tools would start with the same initial abstraction
\sys{lcegar} might not refine at all while \oursolver\ might do. The
reason is that \sys{lcegar} does not abstract all the CFLs which
forced us try with both pair orderings. On the other hand, this gives
some unpredictability to \sys{lcegar} because depending on the
ordering, the tool can behave very differently (for example, 
($C_2 \cap C_4$) versus ($C_4 \cap C_2$)).

\section{Conclusions and Future Work}
\label{sec-conclusion}
We have presented a CEGAR-based semi-decision procedure for
regular separability of context-free languages. We have described
two refinement strategies; an inexpensive greedy approach, and a
more expensive exhaustive strategy.
We have implemented these approaches in a prototype solver,
\sys{covenant}.
The method outperforms \sys{lcegar} on a range of verification
and language-theoretic instances.
The greedy approach often requires more refinement steps, but
tends to quickly find witnesses in cases with non-empty intersections;
the exhaustive method performs substantially more expensive refinement steps,
but can prove separation of some instances not solved by other methods.
  
The maximum $\epsilon$-generalization algorithm can become extremely
expensive for large witnesses. It would be fruitful to consider
whether we can find a cheaper generalization which still ensures completeness.
Similarly, it may be possible to develop a specialized intersection
algorithm for computing $\epsilon$-generalizations, rather than relying
on the standard regular/context-free intersection algorithm.

Visibly pushdown languages (VPLs)~\cite{AlurM04} have become popular
and possess closure and decidability properties very similar to those 
of the class of regular languages.
It would be interesting to explore algorithms for approximation by VPLs.
(Of the languages $C_1$--$C_8$ in Section~\ref{sec-results},
only $C_4$ is a VPL.)

\subsection*{Acknowledgments}
We wish to thank Georgel Calin for providing the test programs and the
implementation of \sys{lcegar}. 
We also thank Pierre Ganty for fruitful discussions about this topic. 
We acknowledge support of the Australian Research Council through 
Discovery Project Grant DP140102194.

\section*{References}

\bibliography{ref}

\begin{thebibliography}{10}

\bibitem{AlurM04}
Rajeev Alur and P.~Madhusudan.
\newblock Visibly pushdown languages.
\newblock In {\em Proceedings of the 36th Annual {ACM} Symposium on the Theory
  of Computing}, pages 202--211. ACM Publ., 2004.

\bibitem{BouajjaniET03}
Ahmed Bouajjani, Javier Esparza, and Tayssir Touili.
\newblock A generic approach to the static analysis of concurrent programs with
  procedures.
\newblock In {\em Proceedings of the 30th Annual {SIGPLAN-SIGACT} Symposium on
  Principles of Programming Languages}, pages 62--73. ACM Publ., 2003.

\bibitem{brzozowski69-stardot}
Janusz~A. Brzozowski and Rina~S. Cohen.
\newblock On decompositions of regular events.
\newblock {\em Journal of the ACM}, 16(1):132--144, 1969.

\bibitem{Chaki_TACAS06}
S.~Chaki, E.~Clarke, N.~Kidd, T.~Reps, and T.~Touili.
\newblock Verifying concurrent message-passing {C} programs with recursive
  calls.
\newblock In H.~Hermanns and J.~Palsberg, editors, {\em Tools and Algorithms
  for the Construction and Analysis of Systems}, volume 3920 of {\em Lecture
  Notes in Computer Science}, pages 334--349. Springer, 2006.

\bibitem{ClarkeGJLV00}
Edmund~M. Clarke, Orna Grumberg, Somesh Jha, Yuan Lu, and Helmut Veith.
\newblock Counterexample-guided abstraction refinement.
\newblock In E.~A. Emerson and A.~P. Sistla, editors, {\em Computer Aided
  Verification}, volume 1855 of {\em Lecture Notes in Computer Science}, pages
  154--169. Springer, 2000.

\bibitem{EsparzaR97}
Javier Esparza and Peter Rossmanith.
\newblock An automata approach to some problems on context-free grammars.
\newblock In C.~Freksa, M.~Jantzen, and R.~Valk, editors, {\em Foundations of
  Computer Science: Potential, Theory, Cognition}, volume 1337 of {\em Lecture
  Notes in Computer Science}, pages 143--152. Springer, 1997.

\bibitem{EsparzaRS00}
Javier Esparza, Peter Rossmanith, and Stefan Schwoon.
\newblock A uniform framework for problems on context-free grammars.
\newblock {\em Bulletin of the EATCS}, 72:169--177, 2000.

\bibitem{revenant}
Graeme Gange, Jorge~A. Navas, Peter~J. Stuckey, Harald S{\o}ndergaard, and
  Peter Schachte.
\newblock Unbounded model-checking with interpolation for regular language
  constraints.
\newblock In N.~Piterman and S.~Smolka, editors, {\em Tools and Algorithms for
  the Construction and Analysis of Systems}, volume 7795 of {\em Lecture Notes
  in Computer Science}, pages 277--291. Springer, 2013.

\bibitem{GrafS97}
Susanne Graf and Hassen Sa{\"{\i}}di.
\newblock Construction of abstract state graphs with {PVS}.
\newblock In {\em CAV}, pages 72--83, 1997.

\bibitem{strsolve}
Pieter Hooimeijer and Westley Weimer.
\newblock {StrSolve}: Solving string constraints lazily.
\newblock {\em Automated Software Engineering}, 19(4):531--559, 2012.

\bibitem{hunt82-grammar}
H.~B. Hunt, III.
\newblock On the decidability of grammar problems.
\newblock {\em Journal of the ACM}, 29(2):429--447, 1982.

\bibitem{bluetooth}
Nicholas Kidd.
\newblock {B}luetooth protocol.
\newblock \url{http://pages.cs.wisc.edu/~kidd/bluetooth}.

\bibitem{LongCMM12}
Zhenyue Long, Georgel Calin, Rupak Majumdar, and Roland Meyer.
\newblock Language-theoretic abstraction refinement.
\newblock In J.~de~Lara and A.~Zisman, editors, {\em Fundamental Approaches to
  Software Engineering}, volume 7212 of {\em Lecture Notes in Computer
  Science}, pages 362--376, 2012.

\bibitem{nagy04-regular}
Benedek Nagy.
\newblock A normal form for regular expressions.
\newblock In {\em Supplemental Papers for the Eighth International Conference
  on Developments in Language Technology}, CDMTCS Research Report Series, pages
  53--62. CDMTCS, 2004.

\bibitem{Nederhof00}
Mark-Jan Nederhof.
\newblock Practical experiments with regular approximation of context-free
  languages.
\newblock {\em Computational Linguistics}, 26(1):17--44, 2000.

\bibitem{Nederhof_chapter1}
Mark-Jan Nederhof.
\newblock Regular approximation of {CFL}s: A grammatical view.
\newblock In H.~Bunt and A.~Nijholt, editors, {\em Advances in Probabilistic
  and Other Parsing Technologies}, volume~16 of {\em Text, Speech and Language
  Technology}, pages 221--241. Springer, 2000.

\bibitem{kiss}
Shaz Qadeer and Dinghao Wu.
\newblock {KISS:} keep it simple and sequential.
\newblock In {\em {PLDI}}, pages 14--24, 2004.

\bibitem{SzymanskiW73}
Thomas~G. Szymanski and John~H. Williams.
\newblock Non-canonical parsing.
\newblock In {\em Conference Record of the 14th Annual Symposium on Switching
  and Automata Theory}, pages 122--129. IEEE Comp.\ Soc., 1973.

\end{thebibliography}
\bibliographystyle{plain} 

\appendix

\newcommand{\stronglyreg}{\mbox{\textsc{StronglyRegGrammar}}}
\newcommand{\makefa}{\mbox{\textsc{MakeFA}}}
\newcommand{\sccs}{\mbox{$SCC$}}
\newcommand{\alphafunction}{\mbox{\textsc{NederhofApproximation}}}

\newpage
\section{Regular Abstractions of Context-Free Grammars}
\label{sec-abstract}

Nederhof~\cite{Nederhof00,Nederhof_chapter1} proposed a transformation
that converts a context-free grammar into a finite automaton. We
present the algorithm \alphafunction\ in Figure~\ref{pcode:approxCFG}
that performs the whole transformation in two steps. The first,
\stronglyreg, describes how to convert a context-free grammar
into a strongly regular grammar. The second step, \makefa, builds a
finite-state automaton from a given strongly regular grammar.

\begin{figure}[h]
  \centerline{
    \pcode{
      $\alphafunction(G)$ \\
      \> let $\sccs$ be the strongly connected components of $G$ \\
      \> $G' \equiv (V',\Sigma,P',S') $ :=  \stronglyreg($G,\sccs$) \\
      \> $Q = \{q_0, q_F\}$ where $q_0$ and $q_F$ are new fresh states; \\
      \> $F$ := $\{ q_F \}$; $\Delta$ := $\emptyset$ \\
      \> \makefa($q_0$, $S'$, $q_F$, $G'$, \sccs) \\
      \> \textbf{return} $R \equiv (Q,\Sigma,\Delta, q_0, F)$ 
    }
  }     
  \caption{\label{pcode:approxCFG}
    Approximating a context-free grammar $G$ with a finite automaton
    $R$ such that $\lang{G} \subseteq \lang{R}$}
\end{figure}

Prior to describing the procedures \stronglyreg\ and \makefa\ we present
some useful definitions.
Consider partitions of the set of non-terminal symbols $V$. We define
$A$ and $B$ to be part of the same partition if $A$ and $B$ are
\emph{mutually recursive}. That is,
$A \Rightarrow^* \alpha B \beta$ \ and $B \Rightarrow^* \alpha' A \beta'$
for some sentential forms $\alpha$, $\beta$, $\alpha'$ and $\beta'$.

\begin{definition}[Left- and Right-Linearity] \rm
A production is \emph{left-linear} iff it is of the form 
$A \rightarrow B\ w$ or $A \rightarrow w$, where $w \in \Sigma^*$.
It is \emph{right-linear} iff it is of the form $A \rightarrow w\ B$
or $A \rightarrow w$, where $w \in \Sigma^*$.
\end{definition}

\begin{definition}[Strongly regular grammars]  
A \emph{strongly regular grammar} is a grammar in which the
productions are either all left-linear or all right-linear.
\end{definition}

\begin{definition}[Left- and Right-Generating]
A set of mutually recursive nonterminals $S$ is \emph{left (right)
generating} if there exists a grammar production $A \rightarrow\ \alpha
B \gamma , \alpha \neq \epsilon$ ($A \rightarrow\ \alpha B \gamma
, \gamma \neq \epsilon$) and $A \in S$.

\noindent Assuming that we have a strongly regular grammar $G$ we can
classify each mutually recursive set $S$ as\footnote{The case where
$S$ is both ``left'' and ``right'' is not applicable here since the
grammar $G$ is strongly regular}:

\vspace{2mm}
\begin{tabular}{ll}
``left'' & if $S$ is not left and right generating \\
``right'' & if $S$ is left and not right generating \\
``cyclic'' & if $S$ is neither left nor right generating \\
\end{tabular}
\end{definition}

\setcounter{proglineno}{0}
\begin{figure}[t]
  \centerline{
    \pcode{
      \stronglyreg($\langle V,\Sigma,P,S\rangle $, $\sccs$) \\
      \putno \> $V'$ := $V$; $S'$ := $S$ ; $P'$ := $\emptyset$ \\
      \putno \> \textbf{foreach} $ i \in \{ 0, \ldots, |SCC| \}$ \\
      \putno \> \> \textbf{if} there exists a production with lhs in $\sccs_i$ neither left- nor right- linear  \\    
      \putno \> \> \> \textbf{foreach} $A \in \sccs_i$ \textbf{do} \\
      \putno \> \> \> \> create a fresh nonterminal $A' \notin V$ \\
      \putno \> \> \> \>  $V'$ := $V' \cup \{A'\}$  \\
      \putno\label{convertStrong-epsilon} \> \> \> \>  $P'$ := $P' \cup \{A' \rightarrow \epsilon \}$ \\
      \putno \> \> \> \textbf{foreach}\ $A \rightarrow \alpha \in P$ with $\alpha \in
      (\Sigma \cup (V \setminus \sccs_i))^{*}$ \textbf{do} \\
      \putno\label{convertStrong-breakcycle-base} \> \> \> \> $P'$ := $P' \cup \{ A \rightarrow \alpha A' \} $  \\
      \putno \> \> \> \textbf{foreach}\ $A \rightarrow \alpha_0 B_1 \alpha_1 B_2
      \alpha_2 \cdots B_m \alpha_m \in P$ with $m \geq 0$ \\
      \putno \> \> \> \hspace{11mm} $B_1,\ldots,B_m \in \sccs_i$, and $\alpha_0,\ldots,\alpha_m \in
      (\Sigma \cup (V \setminus \sccs_i))^{*}$ \textbf{do} \\
      \putno\label{convertStrong-breakcycle-begin} \> \> \> \> $P'$ := $P' \cup\{ $ $A \rightarrow \alpha_0 B_1$  \\
      \putno \> \> \> \> \> \> \> \> $B_1' \rightarrow \alpha_1 B_2$ \\
      \putno \> \> \> \> \> \> \> \> $B_2' \rightarrow \alpha_2 B_3$ \\
      \putno \> \> \> \> \> \> \> \> $\ldots$  \\
      \putno \> \> \> \> \> \> \> \> $B_{m-1}' \rightarrow \alpha_{m-1} B_m$ \\
      \putno\label{convertStrong-breakcycle-end} \> \> \> \> \> \> \> \> $B_{m}' \rightarrow \alpha_{m} A' \}$ \\
      \putno \> \> \textbf{else} \\
      \putno\label{convertStrong-nonrecursive} \> \> \> \textbf{foreach} $A \in \sccs_i$ \textbf{and} $A \rightarrow X \in P$ \textbf{do} $P'$ := $P' \cup \{A \rightarrow X \}$ \\
      \putno \> \textbf{return} $(V', \Sigma, P', S')$
    }
  }
  \caption{\label{pcode:convertStrong}
    Converting a context-free grammar into a strongly regular grammar
  }
\end{figure}

\noindent
Figure~\ref{pcode:convertStrong} shows the transformation,
suggested by Nederhof, to convert an arbitrary context-free grammar
to a strongly regular grammar. 
This transformation is based on the following observation. A
context-free grammar consisting of productions of the form $A
\rightarrow^* \alpha A \beta$ with both $\alpha,\beta$ non-empty
might not be represented as a strongly regular grammar. The intuition
is that $\alpha$ and $\beta$ might be related through an ``unbounded''
communication (\ie some correlation between the occurrences of a's and
b's) no expressible by regular languages.

The procedure \stronglyreg\ takes a context-free grammar $G$
and \sccs, the set of strongly connected components that identify the
set of mutually recursive nonterminal symbols of $G$. The algorithm
identifies in a conservative manner situations where an unbounded
communication can arise and breaks it by adding instead either left or
right linear productions.

The procedure iterates over all strongly connected components in an
arbitrary order and checks if all productions of the nonterminals in
each component are either left or right linear. If yes, no
transformation needs to be applied
(line~\ref{convertStrong-nonrecursive}). Otherwise, it applies some
transformations in order to convert all productions in the same
strongly connected component into either left or right linear, as
described at lines \ref{convertStrong-epsilon} and
\ref{convertStrong-breakcycle-begin}-\ref{convertStrong-breakcycle-end}.
The transformation assumes that nonterminals that do not belong to
$\sccs_i$ are considered as terminals here, for determining if a
production of $\sccs_i$ is right-linear or left-linear. This allows
traversing the strongly connected components in any order.
We show next how the transformation works through an example.

\begin{example}[Conversion to strongly regular grammar] \rm

Let $ G = ( \{\ra,\rb,\rc\}, \{A,B \}, A, P )$,
where $P$ is the set of productions:

\vspace{2mm}
\begin{tabular}{clcl}
 (P1) ~~ & $A$ & $\rightarrow$ & $\ra B \rb$ \\
 (P2) ~~ & $A$ & $\rightarrow$ & $\rc$ \\ 
 (P3) ~~ & $B$ &  $\rightarrow$ & $A$ 
\end{tabular}
\vspace{2mm}

\noindent 
The set of strongly connected components is $\{\{A,B\}\}$. Therefore
there is only one set $\{A,B\}$ of mutually recursive nonterminals.
Line~\ref{convertStrong-epsilon} adds to $P'$ the rules:

\vspace{2mm}
\begin{tabular}{lcl}
 $A'$ & $\rightarrow$ & $\epsilon$ \\
 $B'$ & $\rightarrow$ & $\epsilon$ \\ 
\end{tabular}
\vspace{2mm}

\noindent 
From P1 and executing
lines~\ref{convertStrong-breakcycle-begin}-\ref{convertStrong-breakcycle-end}
we add to $P'$ the rules:

\vspace{2mm}
\begin{tabular}{lcl}
 $A$ & $\rightarrow$ & $\ra B$ \\
 $B'$ & $\rightarrow$ & $\rb A'$ \\
\end{tabular}
\vspace{2mm}

\noindent 
From P2 and executing line~\ref{convertStrong-breakcycle-base} we add
to $P'$ the rule:

\vspace{2mm}
\begin{tabular}{lcl}
 $A$ & $\rightarrow$ & $\rc A'$ \\
\end{tabular}
\vspace{2mm}

\noindent 
From P3 and executing
lines~\ref{convertStrong-breakcycle-begin}-\ref{convertStrong-breakcycle-end}
we add to $P'$ the rule:

\vspace{2mm}
\begin{tabular}{lcl}
 $B$ & $\rightarrow$ & $ A$ \\
 $A'$ & $\rightarrow$ & $B'$ \\
\end{tabular}
\vspace{2mm}

\noindent Finally, putting all rules of $P'$ together and after some
trivial simplifications:

\vspace{2mm}
\begin{tabular}{lcl}
 $A$ & $\rightarrow$ & $\ra A \mid \rc A'$ \\
 $A'$ & $\rightarrow$ & $\epsilon \mid \rb A'$ \\
\end{tabular}
\vspace{2mm}

\noindent  The resulting strongly regular grammar is $G' =
( \{\ra,\rb,\rc\}, \{A,A' \}, A, P' )$.  
Note that $L(G') = \ra^*\rc \rb^*$ while
$L(G) = \{\ra^n \rc \rb^n \mid n \geq 0\}$.  
Therefore, it is easy to see that $L(G) \subseteq L(G')$. 
The transformation broke the synchronization between the number of 
$\ra$'s and $\rb$'s (\#$\ra$'s = \#$\rb$'s) by allowing an
arbitrary number of $\ra$'s and $\rb$'s. \qed
\end{example}

\setcounter{proglineno}{0}
\begin{figure}[t]
  \centerline{
    \pcode{
      \makefa($q_0,A,q_1,G \equiv (V,\Sigma,P,S), \sccs$) \\
      Global variables: set of states $Q$, transition relation $\Delta$ \\ \\
      \putno\label{makeDFA:term-1} \> \textbf{if} $A= \epsilon$ \textbf{then} $\Delta \leftarrow \Delta \cup \{ (q_0,\epsilon,q_1) \}$ \\
      \putno\label{makeDFA:term-2} \> \textbf{else if} $A = a \in \Sigma$ \textbf{then} $\Delta \leftarrow \Delta \cup \{ (q_0,a,q_1) \}$ \\
      \putno\label{makeDFA:descend-beg} \> \textbf{else if} $A = X Y, X \in V, Y \in V^{+}$ \textbf{then} \\
      \putno \> \> $ Q = Q \cup \{q\}$ where $q$ is a new fresh state; \\
      \putno \> \> \makefa($q_0, X, q, G,\sccs$); \\
      \putno\label{makeDFA:descend-end} \> \> \makefa($q, Y, q_1, G,\sccs$) \\
      \putno \> \textbf{else} // $A$ is a nonterminal \\
      \putno \> \> \textbf{if} $A \in \sccs_i$ \textbf{and} $|\sccs_{i}| > 1$ \textbf{then} \\ 
      \putno \> \> \> \textbf{if} $\sccs_i$ is ``left'' \textbf{then} \\
      \putno\label{makeDFA:left-loop-1-beg} \> \> \> \> \textbf{foreach} $C \rightarrow X_1 \cdots X_m \in P$ 
      such that $C \in \sccs_i$ \textbf{and} $X_1 \cdots X_m \notin \sccs_i$ \\
      \putno \> \> \> \> \> $q_C = \mathsf{lookup}(C,Q)$; \\
      \putno\label{makeDFA:left-loop-1-end} \> \> \> \> \> \makefa($q_0,X_1 \cdots X_m, q_C, G, \sccs$) \\
      \putno\label{makeDFA:left-loop-2-beg}\> \> \> \> \textbf{foreach} $C \rightarrow D X_1 \cdots X_m \in P$ 
      such that $C,D \in \sccs_i$ \textbf{and} $X_1 \cdots X_m \notin \sccs_i$ \\
      \putno \> \> \> \> \> $q_C = \mathsf{lookup}(C,Q)$; $q_D = \mathsf{lookup}(D,Q)$; \\
      \putno\label{makeDFA:left-loop-2-end} \> \> \> \> \> \makefa($q_D,X_1 \cdots X_m, q_C, G, \sccs$) \\
      \putno \> \> \> \> $\Delta = \Delta~\cup~\{ (q_A, \epsilon, q_1) \}$ \\
      \putno \> \> \> \textbf{else} // $\sccs_i$ is either ``right'' or ``cyclic''  \\
      \putno \> \> \> \> \textbf{foreach} $C \rightarrow X_1 \cdots X_m \in P$ 
      such that $C \in \sccs_i$ \textbf{and} $X_1 \cdots X_m \notin \sccs_i$ \\
      \putno \> \> \> \> \> $q_C = \mathsf{lookup}(C,Q)$; \\
      \putno \> \> \> \> \> \makefa($q_C,X_1 \cdots X_m, q_1, G, \sccs$) \\
      \putno \> \> \> \> \textbf{foreach} $C \rightarrow X_1 \cdots X_m D \in P$ 
      such that $C,D \in \sccs_i$ \textbf{and} $X_1 \cdots X_m \notin \sccs_i$ \\
      \putno \> \> \> \> \> $q_C = \mathsf{lookup}(C,Q)$; $q_D = \mathsf{lookup}(D,Q)$; \\
      \putno \> \> \> \> \> \makefa($q_C,X_1 \cdots X_m, q_D, G, \sccs$) \\
      \putno \> \> \> \> $\Delta = \Delta~\cup~\{ (q_0, \epsilon, q_A) \}$ \\
      \putno \> \> \textbf{else} \\
      \putno\label{makeDFA:nonrec-nonterminal-beg} \> \> \> \textbf{foreach} $A \rightarrow X \in P$ \textbf{do}  \\
      \putno\label{makeDFA:nonrec-nonterminal-end} \> \> \> \> \makefa($q_0, X, q_1, G, \sccs$) 
    }
  }
  \caption{\label{pcode:makeFA}
    Converting a strongly regular grammar into a finite automata)
  }
\end{figure}

Next, we show in Figure \ref{pcode:makeFA} the procedure to convert a
strongly regular grammar to a finite-state automaton.  The algorithm
is presented as described in~\cite{Nederhof00}.  The main
procedure \makefa\ takes five inputs: the initial state, the string,
the final state reached on reading the string, the strongly regular
grammar, and the mutually recursive nonterminals. We assume a helper
function \textsf{lookup} that maps nonterminals to automata states.
If the nonterminal $A$ is not in the map then a new fresh state $q$ is
returned and the pair $(A,q)$ is inserted in the map. Moreover, we add
$q$ into $Q$. Otherwise if there exists already a pair $(A,q)$ then
$q$ is returned. Note that since \textsf{lookup} can add new states
into $Q$ we pass it as an argument.

Starting from the start symbol of the grammar, the procedure \makefa\
descends the grammar (by triggering
lines~\ref{makeDFA:descend-beg}-\ref{makeDFA:descend-end}) until
terminals are found (lines~\ref{makeDFA:term-1}
and \ref{makeDFA:term-2}), and it creates automata transitions
labelled with those terminals. While descending if it encounters a
non-recursive nonterminal $A$ it continues recursively for each
right-hand side of a production for which $A$ is the left-hand side
(lines~\ref{makeDFA:nonrec-nonterminal-beg}-\ref{makeDFA:nonrec-nonterminal-end}).
The nontrivial case is when the algorithm encounters a recursive
nonterminal symbol $A$. We describe the case when the set of mutually
recursive nonterminals where $A$ belongs to is classified as
``left''. The other case when it is either ``right'' or ``cyclic'' is
symmetric.

Since the set of mutually recursive nonterminals $SCC_i$ is
``left'' all its productions must be of the form (1)
$C \rightarrow \alpha$ or (2) $C \rightarrow D \alpha$, where $C,D \in
SCC_i$ and $\alpha \in (\Sigma \cup (V \setminus SCC_i))^{+}$.
The \textbf{foreach} loop at
lines~\ref{makeDFA:left-loop-1-beg}-\ref{makeDFA:left-loop-1-end}
covers case (1) by descending recursively the right-hand side of the
productions whose left-hand side is denoted by $C$ (\ie $\alpha$) and
passing to the recursive call $q_C$ as final state, an automata state
assigned to $C$. The next \textbf{foreach} loop at
lines~\ref{makeDFA:left-loop-2-beg}-\ref{makeDFA:left-loop-2-end}
handles the other case (2) in a similar manner but the first symbol
$D$ appearing at the right-hand side is another nonterminal symbol
from the same strongly connected component than $C$. We descend again
the right-hand side but providing to the recursive call $q_D$ and
$q_C$ (states assigned to $D$ and $C$) as the initial and final
states, respectively.

\newpage
\section{Intersection of a Context-Free and a Regular Language}
\label{sec-pre-star}

The algorithm to check intersection between a CFL and a regular
language works on context-free grammars with rules of these forms:
\begin{equation}
 A \rightarrow BC \mid a \mid B \mid \epsilon 
\label{eq:normalization}
\end{equation}

\noindent An arbitrary CFG can be converted to a grammar of this form by a
linear increase in terms of the size of the original grammar. Let
$\Rightarrow^{*}$ be the reflexive transitive closure of
$\Rightarrow$. For a grammar $G$ in the above form and a set $L
\subseteq \Sigma_V^*$ we denote by $\prestar{G}{L}$ the set of
predecessors of elements in $L$ with respect to the grammar $G$.  
That is,
\[ 
  \prestar{G}{L} = \{ \alpha \in \Sigma_V^* \mid \exists\ 
	\alpha' \in L: \alpha \Rightarrow^{*} \alpha' \}
\]
Starting from a finite-state automaton $A = \langle Q, \Sigma, \Delta,
q_0, F \rangle$ that recognises a language $L$, and a context-free
grammar $G = (V, \Sigma, P, S)$, we produce an automaton
$A_{\prestar{G}{L}}$ recognizing $\prestar{G}{L}$ by adding transitions to $A$
according to the following saturation rule:

\begin{center}
If $A \rightarrow \beta\ \in P$ and $\langle q, \beta, q' \rangle
\in \Delta$ in the current automaton, add $\langle q, A, q'
\rangle$ to $\Delta$.
\end{center}

\noindent Given the specific form of our input context-free grammar,
we can split the saturation rule into four cases:

\begin{enumerate}

\item for each rule $A \rightarrow \epsilon \in P$ and for each $q_i
  \in Q$ add $\langle q_i, A , q_i \rangle$ to $\Delta$.

\item for each rule $A \rightarrow a \in P$, if there exists
  $\langle q_0,a,q_1\rangle \in \Delta$, then add
  $\langle q_0, A, q_1 \rangle$ to $\Delta$.

\item for each rule $A \rightarrow B \in P$ ($B \in V$), if there exists
  $\langle q_0, B, q_1 \rangle \in \Delta$ then add $\langle q_0, A,
  q_1\rangle$ to $\Delta$.

\item for each rule $A \rightarrow BC \in P$, if there exists
  $\langle q_0, B, q_1 \rangle \in \Delta$ and $\langle q_1, C, q_2
  \rangle \in \Delta$ then add $\langle q_0, A, q_2\rangle$ to
  $\Delta$.

\end{enumerate}

\begin{figure}[b]
\begin{tabular}{cc}
\begin{minipage}{0.35\linewidth}
\centerline{
  \begin{tikzpicture}[->,>=stealth',scale=0.7,every node/.style={scale=0.7},shorten >=1pt,auto,node distance=1.5cm,semithick]
    \newState{q0}{$q_0$}{initial, initial text={}}{minimum size=1pt}
    \newState{q1}{$q_1$}{right of=q0}{minimum size=1pt}
    \newState{q2}{$q_2$}{right of=q1}{minimum size=1pt}
    \newState{q3}{$q_3$}{right of=q2}{minimum size=1pt}
    \newState{q4}{$q_4$}{right of=q3}{accepting,minimum size=1pt} 
    \newTransition{q0}{q1}{$a$}{}
    \newTransition{q1}{q2}{$b$}{}
    \newTransition{q2}{q3}{$b$}{}
    \newTransition{q3}{q4}{$a$}{}
  \end{tikzpicture}}
\end{minipage} &
\begin{minipage}{0.63\linewidth}
\centerline{
  \begin{tikzpicture}[->,>=stealth',shorten >=1pt,auto,scale=0.7,every node/.style={scale=0.7},node distance=2.7cm,semithick]
    \newState{q0}{$q_0$}{initial, initial text={}}{minimum size=1pt}
    \newState{q1}{$q_1$}{right of=q0}{minimum size=1pt}
    \newState{q2}{$q_2$}{right of=q1}{minimum size=1pt}
    \newState{q3}{$q_3$}{right of=q2}{minimum size=1pt}
    \newState{q4}{$q_4$}{right of=q3}{accepting,minimum size=1pt} 
    \newTransition{q0}{q1}{$a|A|B|C$}{}
    \newTransition{q1}{q2}{$b|A|D|E$}{}
    \newTransition{q2}{q3}{$b|A|D|E$}{}
    \newTransition{q3}{q4}{$a|A|B|C$}{}
    \newTransition{q0}{q0}{$A$}{loop above}
    \newTransition{q1}{q1}{$A$}{loop below}
    \newTransition{q2}{q2}{$A$}{loop above}
    \newTransition{q3}{q3}{$A$}{loop below}
    \newTransition{q4}{q4}{$A$}{loop above}
    \newTransition{q0}{q2}{$E$}{bend left}
    \newTransition{q2}{q4}{$B$}{bend left}
    \newTransition{q1}{q3}{$A|E$}{bend right=60}
    \newTransition{q1}{q4}{$B$}{bend right=60}
    \newTransition{q0}{q4}{$A$}{bend right=60}
  \end{tikzpicture}}
\end{minipage} \\
(a) & (b) 
\end{tabular}
\caption{Two finite automata recognizing $L=abba$ and
  $\prestar{G'}{L}$, respectively. The transition $\langle q_0,A,q_4
  \rangle$ means that $abba \in L(G')$.}          
\label{ex:prestar-dfas-2}   
\end{figure}
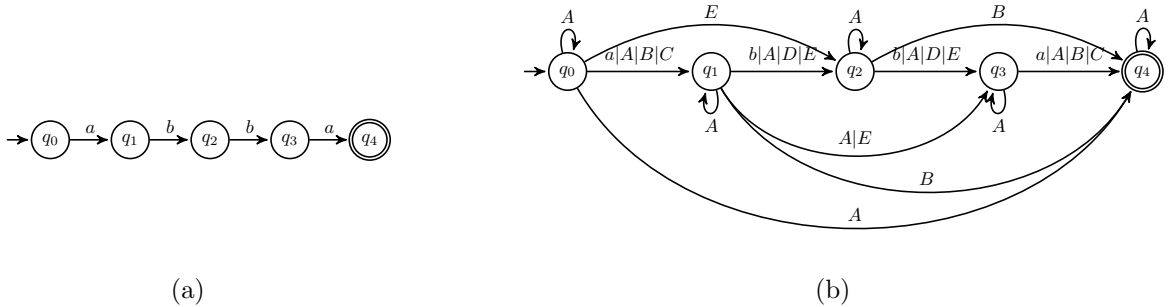

\noindent
The intersection of a regular language and a context-free grammar
can be found by computing first $\prestar{G}{L}$ and then checking if
there is a transition $\langle q_0, S, q_F \rangle \in
\Delta_{A_{\prestar{G}{L}}}$.
The latter check can be done in constant time so the complexity is the
same as the complexity of the $pre^{*}$ algorithm.

\begin{example} \rm
Consider the palindrome grammar $G = \langle \{a,b\}, \{A\}, A, P \rangle$
where $P$ is the set of productions: $A \rightarrow aAa \mid bAb \mid a
\mid b \mid \epsilon $, and the finite automaton shown in
Figure~\ref{ex:prestar-dfas-2}(a) recognizing the language $L=abba$. We
would like to check if $abba \in L(G)$. 

First, we generate the
automaton that recognizes $\prestar{G}{L}$ depicted in
Figure~\ref{ex:prestar-dfas-2}(b) and then we check if there is a
transition from the initial to the final state labelled with $A$ 
(that is, $A \rightarrow^{*} abba$).

In fact, rather
than computing $\prestar{G}{L}$ directly we normalize our input
grammar obtaining $G' = \langle \{a,b\}, \{A,B,C,D,E \}, A, P \rangle$
where $P$ is the set of productions:

\begin{center}
\begin{tabular}{lcl}
 $A$ & $\rightarrow$ & $CB \mid DE \mid a \mid b \mid \epsilon$ \\
 $B$ & $\rightarrow$ & $AC$ 
\end{tabular}
\qquad
\begin{tabular}{lcl}
 $C$ & $\rightarrow$ & $a$ \\
 $D$ & $\rightarrow$ & $b$ 
\end{tabular}
\qquad
\begin{tabular}{lcl}
 $E$ & $\rightarrow$ & $AD$ \\
 ~ & 
\end{tabular}
\end{center}
\noindent and compute $\prestar{G'}{L}$ instead. 
Note that since there
exists a transition from $q_0$ to $q_4$ with label $A$ in the
automaton shown in Figure~\ref{ex:prestar-dfas-2}(b) we have proven
that $abba \in L(G')$.

\end{example}

\newpage
\section{Recursive Multithreaded Programs}
\label{sec-benchmarks}

Detailed descriptions of the programs used in Section~\ref{sec-results}'s
experimental evaluation, as well as their safety properties, can be found
in~\cite{LongCMM12,Chaki_TACAS06,kiss}. This appendix is intended 
as a self-contained short description.

There are two classes of programs: Erlang programs extracted from
textbook algorithms and several variants of a real Bluetooth driver
implementation.
Table~\ref{benchmarks-info} shows the sizes of the programs after each
context-free grammar has been normalized (that is, they satisfy the form
in Equation~\ref{eq:normalization}):

\begin{itemize}
\setlength{\itemsep}{-0.5ex}
\renewcommand{\labelitemi}{\scriptsize$\bullet$} 
\item  $\#CFGs$: the number of context-free grammars
\item  $|\Sigma|$  number of terminal symbols
\item  $|N|$: total number of nonterminal symbols
\item  $|P|$: total number of grammar productions
\end{itemize}

\begin{table}[h]
\begin{center}
\begin{tabular}{|l|c|c|c|c||l|c|c|c|c|}
      \hline
      Program  & $\#CFGs$ &  $|\Sigma|$ & $|N|$ & $|P|$ & Program  & $\#CFGs$ &  $|\Sigma|$ & $|N|$ & $|P|$ \\ 
      \hline
      \hline
      \textsf{SharedMem} & 4 & 8  & 138 & 234 & \textsf{Version 1} &  7 & 17 & 471 & 804 \\
      \hline
      \textsf{Mutex} & 4 & 22  & 297  & 512 & \textsf{Version 2} & 9 & 26 & 1055 & 1847 \\
      \hline
      \textsf{RA} & 2 & 20 & 127 & 205 & \textsf{Version 2 w/ Heuri} & 9 & 26 & 807 & 1351 \\
      \hline
      \textsf{Modified RA} & 5 & 22 & 323 & 530 & \textsf{Version 3 (1A2S)} & 9 & 22 & 746 & 1292 \\ 
      \hline
      \textsf{TNA} & 3 & 17 & 134 & 204 & \textsf{Version 3 (1A2S) w/ Heuri} & 8 & 22 & 569 & 938 \\
      \hline
      \textsf{Banking} & 3  & 13 & 144  & 244  & \textsf{Version 3 (2A1S)} & 9 & 25 & 1053 & 1052 \\
      \hline
\end{tabular}
\caption{\label{benchmarks-info} Sizes of the programs shown in Table~\ref{results}(a-b)}
\end{center}
\end{table}

\noindent \textbf{Erlang programs.} \textsf{SharedMem} is the shared memory
 program shown in detail in Figure~\ref{fig:encoding}. \textsf{Mutex}
is an implementation of the Peterson mutual exclusion protocol where
two processes try to acquire a lock. The checked property is that at
most one process can be in the critical section at any one time.
\textsf{RA} is a resource allocator manager that handles
``allocate'' and ``free'' requests. We check that the manager cannot
allocate more resources to clients than there are currently free 
resources in the system. 
\textsf{Modified RA} adds some new functionality to the
logic of the resource allocator manager. We check the same property
used in \textsf{RA}. \textsf{TNA} is a telephone number analyzer that
serves ``lookup'' and ``add number'' requests. The property to check
is that some programming errors cannot
happen. Finally, \textsf{Banking} is a toy banking application where
users can check a balance as well as deposit and withdraw money. 
We check that deposits and withdrawals of money are done atomically.

\vspace{2mm}
\noindent \textbf{Bluetooth driver~\cite{bluetooth}.}  This is a simplified 
implementation of a Windows NT Bluetooth driver and several variants
discussed originally in~\cite{kiss}. The driver keeps track of how
many threads are executing in the driver. The driver increments
(decrements) atomically a counter whenever a thread enters (exits) the
driver. Any thread can try to stop the driver at any time, and after
that new threads are not supposed to enter the driver. When the driver
checks that no threads are currently executing the driver, a flag is
set to true to establish that the driver has been stopped. Other
threads must assert this flag is false before they start their work in
the driver. There are two dispatch functions that can be executed by
the operative system: one that performs I/O in the driver and another
to stop the driver.  Assuming that threads can asynchronously execute
both dispatch functions we check the following race condition: no
thread can enter in the driver after the driver has been
stopped. \textsf{Version 1} and \textsf{Version 2}~\cite{kiss} are two
buggy versions of the driver implementation. \textsf{Version 2 w/
Heuri} is an alternative encoding of \textsf{Version 2} introduced
by~\cite{LongCMM12} to limit context switches only at basic block
boundaries. This makes the verification task easier but it is, in
general, unsound as it does not cover all possible behaviours of
the driver. \textsf{Version 3 (2A1S)}~\cite{Chaki_TACAS06} is a safe
version after blocking the counterexample found in \textsf{Version 2}
where two adder and one stopper processes are
considered, \textsf{Version 3 (1A2S)} is a buggy version with one
adder and two stopper processes, and finally, \textsf{Version 3 (1A2S)
w/ Heuri} is an alternative encoding with the unsound heuristics used
in \textsf{Version 2 w/ Heuri}.

\end{document}